\documentclass[prl,twocolumn,floatix,amssymb,showpacs,amsmath,superscriptaddress]{revtex4-1}
\usepackage{graphicx}
\usepackage{epsfig,color}
\usepackage{amsmath,amssymb,amsthm}
\usepackage[utf8]{inputenc}
\usepackage[colorlinks, urlcolor=blue, linkcolor=red]{hyperref}
\newtheorem{theo}{Theorem}
\newtheorem{lemma}{Lemma}

\newtheorem{definition}{Definition}

\DeclareMathOperator{\arcsinh}{arcsinh}
\usepackage{soul}

\begin{document}

\title{Two-way covert quantum communication in the microwave regime}
\author{R. Di Candia}
\email{rob.dicandia@gmail.com}
\affiliation{Department of Communications and Networking, Aalto University, Espoo, 02150 Finland}
\author{H. Yi\u{g}itler}
\affiliation{Department of Communications and Networking, Aalto University, Espoo, 02150 Finland}
\author{G. S. Paraoanu}
\affiliation{QTF Centre of Excellence, Department of Applied Physics,
Aalto University School of Science, FI-00076 AALTO, Finland}
\author{R. J\"antti}
\affiliation{Department of Communications and Networking, Aalto University, Espoo, 02150 Finland}

\begin{abstract}

Quantum communication addresses the problem of exchanging information across macroscopic distances by employing encryption techniques based on quantum mechanical laws. 
Here, we advance a new paradigm for secure quantum communication by combining backscattering concepts with covert communication in the microwave regime. 
Our protocol allows communication between Alice, who uses {\it only} discrete phase modulations, and Bob, who has access to cryogenic microwave technology. Using notions of quantum channel discrimination and quantum metrology, we find the ultimate bounds for the receiver performance, proving that quantum correlations can enhance the signal-to-noise ratio by up to $6$~dB. These bounds rule out any quantum illumination advantage when the source is strongly amplified, and shows that a relevant gain is possible only in the low photon-number regime. We show how the protocol can be used for covert communication, where the carrier signal is indistinguishable from the thermal noise in the environment. We complement our information-theoretic results with a feasible experimental proposal in a circuit QED platform. This work makes a decisive step toward implementing secure quantum communication concepts in the previously uncharted $1-10$~GHz frequency range, in the scenario when the disposable power of one party is severely constrained. 

\end{abstract}

\maketitle

\section{I. Introduction}
It is well understood that the application of quantum mechanics to traditional technology-related problems may give a new twist to a number of fields. Quantum communication is a potential candidate for over-passing its classical counterpart in relevant aspects of information-theoretic security. By appropriately encoding the information in the degrees of freedom of quantum systems, a possible eavesdropping attack can be detected due to the sensitivity of the system to the measurement process. This simple reasoning has been at the basis of defining a number of quantum key distribution (QKD) protocols during the first quantum information era, such as BB84~\cite{BB84}, E91~\cite{E91} and B92~\cite{B92}. The defined protocols have been proven to be unconditionally secure provided that the transmitting channel has a low noise~\cite{Shor84,Pirandola19}. The same level of security would be impossible to reach even in the most sophisticated known classical architectures, which rely on the current impossibility of solving efficiently specific problems, such as prime number factorization or finding the solution of systems of multivariate equations~\cite{post}. 
This means that classical encryption techniques are not fundamentally secure:
information considered to be safely stored {\it today} is not guaranteed to be so {\it tomorrow}~\cite{Mosca}. Quantum communication aims to solve this long-term security problem at some infrastructure costs yet to be quantified.

From a theoretical point of view, there is a challenge in defining quantum communication protocols which are secure, efficient and practical at the same time. In this respect, optical systems have been considered for decades the main candidates for quantum communication, as thermal effects are negligible in this frequency range. For instance, QKD security proofs require level of noises which at room temperature are reachable only by frequencies at least in the Terahertz band~\cite{thermalQKD}. In addition, entanglement can be distributed with minimal losses, allowing for the implementation of a series of key long-distance quantum communication experiments, such as quantum teleportation~\cite{Qtel}, device-independent QKD~\cite{DevQKD}, deterministic QKD~\cite{Lo}, superdense coding~\cite{dense}, and continuous-variable quantum information correction schemes~\cite{Sabuncu, Mista, Lassen}, among others. The realization of these experiments have been mainly possible because of large efforts in improving photon-detection fidelities, single-photon generation, high-rate entanglement generation, and on-chip fabrication methods~\cite{onchip}. Despite these advances in optical technology for quantum communication, low-frequency systems, such as those operating in the microwave or radio wavelengths, have still advantages related to easier electronic design. In addition, microwave signals in the range $100$~MHz -- $10$~GHz belong to the low-opacity window, therefore are particularly suitable for open-air communication applications.  In fact, one may think as optical fibers as a better choice for long-distance communication, due to their resilience to environmental dissipation, and microwave systems for short- and medium- distance applications for the low-cost of their electronics~\cite{Karsa2020}. It is therefore compelling to investigate at the fundamental level whether secure open-air communication protocols are possible at larger wavelengths, with a long-term idea of reaching a network design integrating quantum and classical links, with minimal possible changes in the already existing infrastructure. 
With the advent of circuit QED (cQED) as a promising platform for quantum computation~\cite{Blais2020}, experimental and theoretical research has been focused on understanding the properties of microwave signals at the quantum level. If cooled down at $20$~mK, thermal effects are suppressed and microwave electromagnetic fields with frequency above a few GHz show exemplary quantum effects, such as superposition, entanglement and squeezing below vacuum~\cite{Menzel12}. Lately, we have witnessed several experimental advances, which can be regarded as milestones for developing microwave quantum communication, such as improved Josephson parametric amplifiers (JPAs)~\cite{Zhong13, Pogorzalek17}, microwave photodetectors~\cite{Bessel18} and bolometers~\cite{Kokkoniemi20}, generation of path-entanglement~\cite{Menzel12, Fedorov18, DiCandia14}, generation of multi-mode entangled states~\cite{Lahteenmaki2016, Wilson2018, Bruschi2017}, remote state preparation~\cite{Pogorzalek19} and quantum teleportation~\cite{Fedorov2021, DiCandia15, Fedorov16}. The ability of generating and manipulating microwave radiation in a way that reproduces quantum optics experiments allows to extend the concept of photon as a single-quanta excitation of the electromagnetic field. The short-term promise in the field is to demonstrate quantum communication and quantum sensing protocols protocols in the microwave regime~\cite{Barzanjeh15, LasHeras17, Jonsson,Shabir,Wilson}, which would then enable real-life applications~\cite{Sanz18}. Recent theoretical results in noisy quantum sensing and metrology show that preserving entanglement in an experiment is not a fundamental feature for reaching a quantum advantage~\cite{Tan08, Braun18}, paving the way for the implementation of open-air quantum microwave protocols.

\begin{figure}[t!]
	\centering
\includegraphics[width=0.45\textwidth]{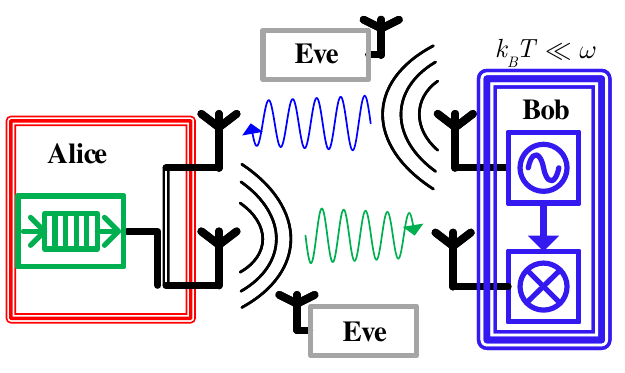}
\caption{{\bf Sketch of the two-way quantum communication protocol.} Bob sends a signal to Alice, who embeds the message by phase modulation. The signal is then sent back to Bob, who retrieve the information via a suitable measurement. Bob may use quantum correlations in order to increase the signal-to-noise ratio. A passive Eve is able to collect the photons lost in both paths, but she does not have access to Bob and Alice labs.}
\label{draw1}
\end{figure} 

This article exploits recent results in quantum communication, quantum sensing and cQED in order to introduce a feasible secure {\it two-way} quantum communication protocol in the microwave regime. The protocol combines microwave quantum radar technology with covert communication. It consists in the secure exchange of information between a classical party (Alice) and a quantum party (Bob), who pre-share a secret. Bob sends a continuous-variable microwave signal to Alice, which encodes her message in the phase modulation according to a pre-agreed alphabet. The signal is then transmitted back to Bob and measured in order to discriminate between the different modulations (see Fig.~\ref{draw1}). As Alice is performing uniquely passive operations at room temperature, she needs only classically available components.
If seen from the energy-expenditure perspective, one can think of {\it one-} and {\it two-way} protocols as having fundamentally different features in quantum communication. In one-way protocols, Alice (the message transmitter) generates quantum states of some sort, while Bob (the message receiver) has access to some operations typically easy to implement, and a measurement device. Taking into account that even the simplest of the detection schemes, such as homodyne or heterodyne, requires amplifiers and signal generators, a non-negligible energy disposal for both Alice and Bob is required in order to implement any one-way protocol. Our two-way protocol, instead, puts all the challenging  technological requirements at Bob's side. This feature has recently gained a lot of interest in the communication engineer community, especially for implementing backscatter communication~\cite{Griffin, Navaz, Niu}. In fact our protocol is significantly different with respect to previous proposals~\cite{Pirandola08,Pirandola2way, Ghorai19, Pirandola1way, Shapiro09}, since it does not require active control at Alice side. Since Alice's energy requirements are minimized, we envision real-life applications in ultra-low power RF communications, Internet-of-Things, and Near Field Communication (NFC) based technology, among others. While the possibility of using this technology has been speculated in the literature, no rigorous approach has ever been pursued. One of the goal of this paper is indeed to provide a solid theoretical basis to the quantum backscatter communication field, together with a feasible implementation proposal. We discuss both the cases when the signal is in a coherent state and when it is correlated with an idler. The latter shows a gain of up to $6$~dB in the signal-to-noise ratio (SNR) with respect to the former one, at some experimental cost in the preparation and the detection stages by Bob. The setup resembles the Gaussian quantum illumination protocol~\cite{Tan08, Zhang13,Zhang15}, where a weak two-mode squeezed vacuum (TMSV) state is transmitted in a bright environment in order to detect the presence or absence of a low-reflectivity object in a region of space. Unlike radar applications, for which the quantum illumination paradigm is usually employed~\cite{Barzanjeh15, Wilson}, and where location, velocity and cross-section are unknown, our communicating setup can be thought to be applied with static antennas where all these properties are known and can be engineered. In the first part of the paper, we derive a general expression for the error probability, putting an emphasis on Gaussian states and Schr\"odinger's cat (SC) states~\cite{Devoret13}. We prove that $6$~dB is indeed the maximal gain in the error probability exponent reachable by a quantum-correlated state over a coherent state receiver. This also settles the ultimate limits of quantum illumination~\cite{Shapiro19}. We show that our communication protocol is covert~\cite{Bash15,Arrazola16, Liu17, Bullock2020}. In a covert quantum communication protocol the signal is hidden in the thermal noise unavoidably present in a room-temperature environment, so that Eve's detection probability collapses. The basic idea is therefore to protect the message content by hiding its existence. This concept has a natural application in low-frequency spectrum communication. Covertness is achievable only if the generated signal is weak enough, so that cryogenic detection technology is needed at Bob's side. This concept can be applied in situations when one does not want to expose the metadata about the propagation channel access time and duration. In addition, since covertness works in the regime where the power of all involved signals is low, one is allowed to use already licensed frequencies without compromising the performance of the licensed users~\cite{Griffin}. In our protocol, we use a one-time-pad cipher to allow for the covert channel to be measured only by the intended recipients (detectability), and to make the transmitted and reflected signals uncorrelated between each other and indistinguishable from thermal noise (indistinguishability). In this way the protocol is also unconditionally secure. We show explicit strategies for how to use entanglement to increase the effective bandwidth, i.e. the number of data hiding bits per channel use.  Indeed, we show the square-root-law for our two-way communication protocol, which states that the number of bits that can be sent over $n$ channel usages scales as $O(\sqrt{n})$, finding explicitly all the multiplicative constants for various transmitters and receivers. In addition, we design an entanglement-assisted protocol based on SC states in a circuit QED (cQED) setup, which relies solely on Jaynes-Cumming interactions and qubit measurements. This is an important requirement for a circuit QED implementation, since a receiver based on photon-detection is currently too demanding to be realistic. This shows the path for future experimental investigation of our concept.

The article is organized in the following way. In Section~\hyperref[sec:II]{II}, we describe our communicating setup. In Section~\hyperref[sec:III]{III}, we provide the ultimate bounds on the receiver performance for both local and collective strategies. In Section~\hyperref[sec:III]{IV}, we focus on the cases of Gaussian and SC states, providing an explicit expression for the Chernoff bound in the corresponding channel discrimination problem. In Section~\hyperref[sec:IV]{V}, we discuss the conditions on the average transmitting power in order to have a covert system, together with a proof of the square-root law for our two-way setup for the discussed transmitters and receivers. We also discuss the key-expansion and synchronization protocols. In Section~\hyperref[sec:V]{VI}, we discuss a cQED protocol based on SC states, with a receiver design based on Jaynes-Cumming interactions and qubit measurements.

\section{II. The Setup}\label{sec:II}

\begin{figure}[t!]
	\centering
\includegraphics[width=0.45\textwidth]{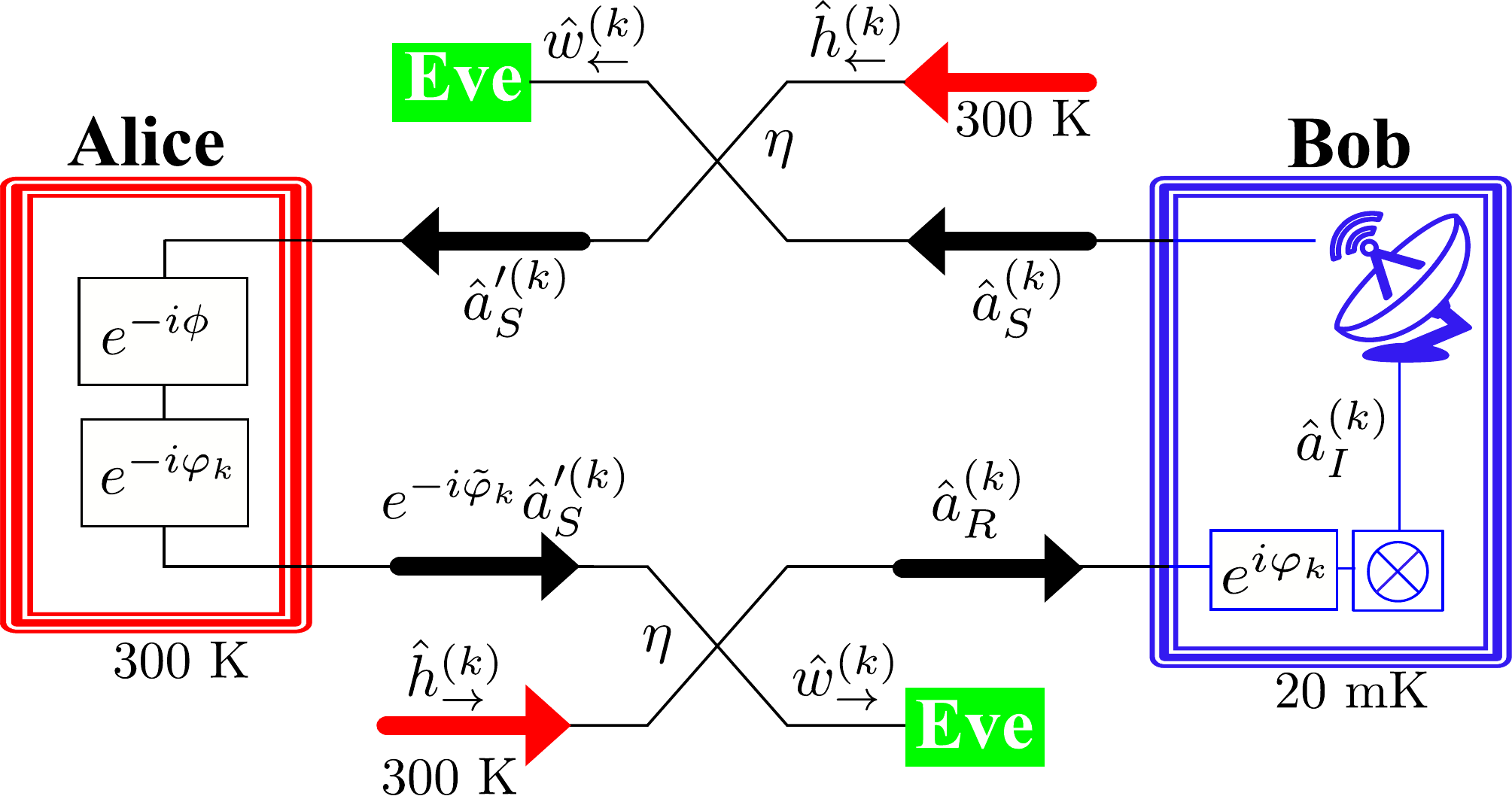}
\caption{{\bf Setup for the two-way covert quantum communication through a bright bosonic channel.} A signal microwave mode $\hat a_S^{(k)}$, possibly entangled with an idler mode $\hat a_I^{(k)}$, is generated by Bob. The idler mode is stored in the lab, while the signal is sent through a noisy channel to Alice, who receive the noisy modes $\hat a'^{(k)}_S$. Alice modulates the phase of the signal by $\tilde \varphi_k=\phi+\varphi_k$, where $\phi$ and $\varphi_k$ belong to a pre-agreed discrete alphabet $\mathcal{A}$. Here, $\phi$ embeds the information to be sent, while $e^{-i\varphi_k}$ is an encoding operator. The value of $\varphi_k$ is taken uniformly at random in $\mathcal{A}$, and it is known only to Alice and Bob. The signal is then scattered back to Bob, who decodes it by applying a phase modulation $e^{i\varphi_k}$. This process is repeated $M$ times ($k=1,\dots,M$), for each symbol transmission. Bob performs a measurement on the modes $\{e^{i\varphi_k}\hat a_R^{(k)}, \hat a_I^{(k)}\}$ in order to discriminate between the possible values $\phi$. Eve performs a collective measurement on the modes $\{\hat w_\leftarrow^{(k)},\hat w_\rightarrow^{(k)} \}$ in order to understand whether Alice and Bob are communicating. If the average power of the signal modes is $O(\eta^2 N_B/\sqrt{n})$, then Alice and Bob are able to use covertly $n$ channel modes. This allows to transmit $O(\sqrt{n})$ number of bits in a secure way. In the $1-10$~GHz band, Bob's signals are generated at $20$~mK in order to suppress the thermal contribution and comply with the covertness conditions.}
\label{Setup}
\end{figure}

The two-way communication protocol is depicted in Fig.~\ref{Setup}. Alice  and  Bob  use $n=mM$ modes  of  the bosonic channel  simultaneously in order to communicate $m$ symbols. They  use $M$ modes  to transmit  a  symbol $\phi$ taken  from  a  discrete  alphabet $\mathcal{A}$. We  refer  to  each these $M$ channel  usages as  a {\it slot}. In each slot,  Bob generates $M$ independent and identically distributed (i.i.d) signal modes $\{\hat a_S^{(k)}\}$ ($k=1,\dots,M$) with $N_S>0$ average number of photons, and he sends the modes to Alice. The signal modes are possibly entangled with $M$ idler modes $\{\hat a_I^{(k)}\}$, which are retained in the lab by Bob for the measurement stage. The signals are generated at a low enough temperature to consider the signal-idler (SI) state as pure. Although the results of this article are general, we emphasize the application in the microwave spectrum, specifically in the range of operating frequencies of a cQED setup, i.e. $1$-$10$ GHz. In this range of frequencies, $T\simeq20$~mK is required to avoid thermal fluctuations. We refer to the {\it idler-free} case when the idler is absent, or, equivalently, when the signal and the idler are uncorrelated.
The signal modes are sent to Alice through a room-temperature channel ($T_B=300$~K), which is modeled as a beamsplitter. Alice receives the modes $\{\hat a'^{(k)}_S\}$, with
\begin{align}
\hat a'^{(k)}_S=\sqrt{\eta}~\hat a_S^{(k)} + \sqrt{1-\eta}~\hat h_\leftarrow^{(k)}.
\end{align}
Here, $\eta$ is the power transmitting rate of the channel and $\{\hat h_\leftarrow^{(k)}\}$ are independent thermal modes with $N_B$ average number of photons. The numerical value of $N_B$ depends on the signal operating frequency $\omega_k$ as $N_B=(e^{\beta\hbar \omega_k}-1)^{-1}$ with $\beta=(k_B T_B)^{-1}$, $k_B$ being the Boltzmann constant. In the $1-10$~GHz spectrum this results to values of the order $N_B\sim 10^3$, therefore we will emphasize the $N_B\gg1$ case. Alice modulates the phase of $\hat a'^{(k)}_S$ by $\tilde\varphi_k=\phi+\varphi_k$, with $\phi,\varphi_k\in\mathcal{A}$, generating the mode $e^{-i\tilde \varphi_k}\hat a'^{(k)}_S$. She then sends the signal back to Bob through the same channel. Here, $\phi$ embeds the symbol to be transmitted, while the phase-shift $e^{-i\varphi_k}$ is an encoding operation that Alice and Bob have secretly pre-shared. In other words, Alice and Bob uses a one-time-pad protocol to encode the symbol $\phi$. Bob receives the modes $\{\hat a_R^{(k)}\}$, with
\begin{align}
\hat a_R^{(k)} =\sqrt{\eta}~ e^{-i\tilde \varphi_k}\hat a'^{(k)}_S + \sqrt{1-\eta}~\hat h_\rightarrow^{(k)},
\end{align}
Here, $\{\hat h_\rightarrow^{(k)}\}$ are $M$ independent thermal modes identical to $\{\hat h_\leftarrow^{(k)}\}$. We also assume that the modes $\{\hat h_\leftarrow^{(k)}\}$ and $\{\hat h_\rightarrow^{(k)}\}$ are independent.
Bob applies the decoding transformation $e^{i\varphi_k}$ to the received mode $\hat a_R^{(k)}$. He then applies a discrimination strategy to the modes $\{e^{i\varphi_k}\hat a_R^{(k)},\hat a_I^{(k)}\}$ for distinguishing between the different symbols in $\mathcal{A}$. 

For a given symbol transmission $\phi$, we denote with $\rho_{\eta,\phi}$ the density matrix of Bob's state at the receiver, i.e. the state of the system defined by the modes $\hat a_R^{(k)}$ and $\hat a_I^{(k)}$. As we are working in the i.i.d. assumption, $\rho_{\eta,\phi}$ does not depend on $k$. In the following, we work under the $\eta\ll1$ assumption, corresponding to a very lossy thermal propagation channel. This is the case of non-directional antenna and/or unfavorable weather conditions. However, the theory can be extended to finite values of $\eta$ as well.  The number $M$ has to be chosen to be large enough in order to give Bob the chance of discriminating between the possible phases in $\mathcal{A}$ with high confidence. The measurement discriminating between the symbols depends on the adopted SI system. We consider mainly the Binary-Phase-Shift-Keying (BPSK) alphabet, when $\mathcal{A}=\{0,\pi\}$. However, we will discuss how  the results in this article can be extended to more complex alphabets. Different figures of merit can be used to quantify the performance of the optimal strategies to discriminate between the distinct modulations, depending on their a priori probabilities. Here, we discuss the case where all the modulations in the key have the same a priori probability of being realized, which is the most natural scenario for quantum communication. 

The setup can be mapped to quantum illumination~\cite{Tan08}, also referred to On-Off-Keying (OOK), where Alice modulates the amplitude of the signal. In fact, BPSK and OOK share the same optimal strategies in the $\eta\ll1$ limit (see Lemma~\ref{lemma1} of the Appendix). In addition, BPSK performs better than OOK for given transmitting power, as the distance of the symbols in the phase-space is larger. 
A similar setup has been studied in the optical domain~\cite{Shapiro09}. Here, a phase-insensitive amplification by Alice is required in order to add thermal noise and ensure security with respect to a passive Eve. In the low-frequency spectrum, the thermal noise is naturally present in the environment, so that no amplification is needed and covertness can be ensured. Instead, in Ref.~\cite{Shi20}, the authors discuss the advantage of using pre-shared entanglement between Alice and Bob for communication in noisy environment, finding that the number of covert bits that can be sent increases by a logarithmic factor with respect to the unentangled case. Although their setup falls in the one-way scenario, these results suggest that using quantum correlations may come with a logarithmic overhead in the capacity also in our case. 

\section{III. Optimal receiver performance}\label{sec:III}
In this section, we find the ultimate bounds on the receiver performance for the protocol described in Fig.~\ref{Setup}. These results hold for SI systems in any quantum state. We set the encoding operation to the identity, i.e. we fix $\varphi_k=0$. This is possible because both Alice and Bob have a pre-shared knowledge of $\varphi_k$, therefore this operation can be reversed by Bob. In the BPSK case, where the $\phi\in\{0,\pi\}$, our aim is to minimize the total error probability 
\begin{equation}
    p_{{\rm err}} = \frac{1}{2}\left[\Pr\left( \phi=\pi|\phi=0\right)+\Pr\left(\phi=0|\phi=\pi\right)\right],
\end{equation}
where $\Pr\left( \phi=a|\phi=b\right)$ is the probability of detecting a phase $\phi=a$ given that Alice has transmitted the symbol $\phi=b$. The main strategies to achieve this can be classified in  (i) {\it collective}, where the $M$ modes are allowed to be measured together, and  (ii) {\it local}, where the copies are measured separately, allowing classical communication between the measurements on the different copies. We consider the performance of a coherent state transmitter as reference for the correlated cases. In other contexts, such as in quantum illumination, coherent states transmitters are usually used as a {\it classical} reference. This choice is done mainly for two reasons: 1) they achieve the optimal error probability in the idler-free case and 2) they describe faithfully coherent signals that can be generated with classical technology.

\begin{figure}[t!]
	\centering
	\includegraphics[width=0.47\textwidth]{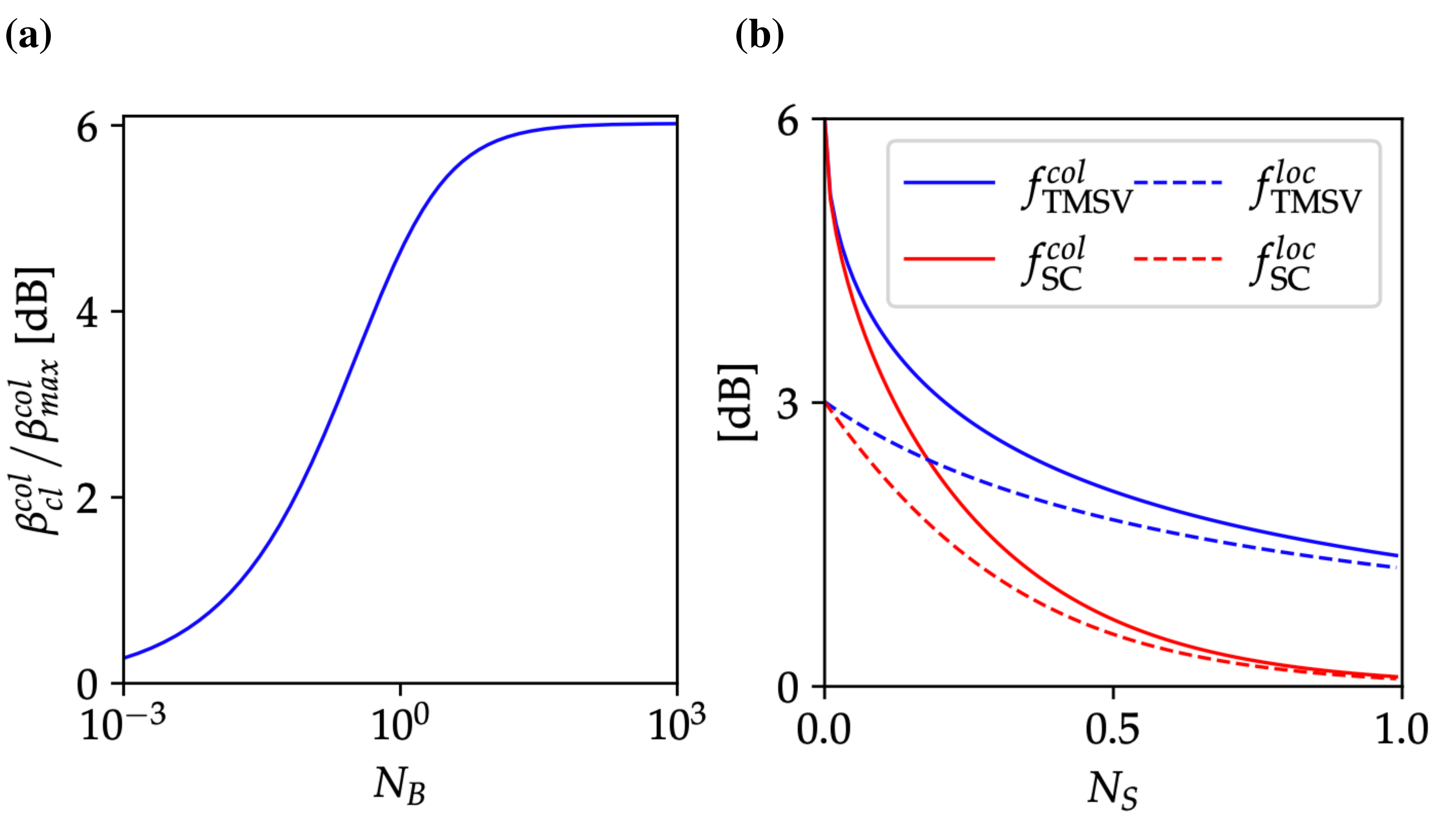}
\caption{{\bf Performance of the quantum correlated protocol with respect to the idler-free case.} {\bf (a)} Maximal achievable gain of the Chernoff bound of the quantum correlated case (denoted as $\beta^{col}_{max}$) with respect to the idler-free case ($\beta^{col}_{cl}$), depending on the average number of thermal photons in the environment $N_B$. For instance, for $N_B=1$, the maximal gain is about $4.6$~dB and it is reached in the $N_S\ll1$ limit. {\bf (b)} Comparison of the optimal receiver performance for Gaussian states and Schr\"odinger's cat states with respect to  the idler-free case, in the $N_B\gg1$ limit. The performance is quantified in terms of the error probability decaying exponent. Indeed, the quantities $g^{col}_{\rm TMSV, SC}=\beta^{col}_{\rm TMSV,SC}/\beta^{col}_{cl}$ and $g^{loc}_{\rm TMSV, SC}=\beta^{loc}_{\rm TMSV,SC}/\beta^{loc}_{cl}$ are plotted for $N_B\gg1$. The graphics shows how the advantage in using quantum correlations decays with the transmitting power $N_S$. Gaussian states perform better than Schr\"odinger's cat states for finite $N_S$.}\label{fig:3}
\end{figure} 

\subsection{A. Collective strategies}
The quantum Chernoff bound~\cite{Audenaert07,Calsamiglia08,PirandolaGaussian} provides an upper bound on the achievable error probability (EP) in the binary detection problem. Indeed, we have that $p_{{\rm err}}\leq \frac{1}{2} e^{-\beta_\eta M}$, where $\beta_\eta = -\min_{s\in(0,1)}\log\text{Tr}\,(\rho_{\eta,0}^{s}\rho_{\eta,\pi}^{1-s})$. This bound is tight for $M\gg1$, and its exponent $\beta_\eta$ can be used as figure of merit for quantifying the performance of the optimal discrimination protocol. Generally, one needs a collective measurement over all the $M$ modes in order to saturate this bound. There are few exceptions to this statement, for instance, when either one of $\rho_{\eta, 0}$ or $\rho_{\eta,\pi}$ is close enough to a pure state~\cite{Calsamiglia08}, or for coherent state illumination in a bright environment. In the following, we focus on computing analytically $\beta_\eta$ up to the first relevant order in $\eta$, using a metric-based approach. Since the expansion of $\beta_\eta$ to the first order of
$\eta$ is zero (Lemma~\ref{Chernoff} of the Appendix), we can define the figure of merit for collective strategies as
\begin{equation}
    \beta^{col} \equiv \lim_{\eta \rightarrow 0} \frac{\beta_\eta}{\eta^2}.
\end{equation}
 This metric provides us a framework for conducting comparative analysis between different transmitters. In addition, the found relations can be analytically computed, providing  an insight on the scaling of the performance with respect to the system parameters. We are particularly interested in the $N_S\ll1$ and $N_B\gg1$ limits of $\beta^{col}$, where the protocol based on quantum correlations will present the maximal advantage with respect to an uncorrelated input state with the same power. In addition, we will see that unconditional security by means of covertness can be ensured in this regime. In the following, for simplicity, we will refer to $\beta^{col}$ as the quantum Chernoff bound (QCB). However, we stress that in the literature the QCB is generally referred as $\beta_\eta$. Let us denote as $\beta^{col}_{\rm free}$ the value of $\beta^{col}$ if the input is idler-free, i.e. when the signal and the idler are uncorrelated.  The following result define the optimal receiver performances.

\begin{theo}{\bf [Ultimate receiver-EP bound]}~\label{QCB}\\
The following bounds holds for the receiver performance with collective strategies:
\begin{align}
\beta^{col}&\leq \min\left\{\frac{4N_S}{1+N_B},\frac{1}{\sqrt{c_B}}\frac{N_S+\frac{1}{2}}{1+N_B}\right\}, \\
\beta_{\rm free}^{col}&\leq\frac{4}{\left(1+\sqrt{c_B}\right)^2}\frac{N_S}{1+N_B}\equiv \beta_{cl}^{col},
\end{align}
where $c_B=\frac{N_B}{1+N_B}$. The idler-free case bound is saturated by  coherent state signals.
\end{theo}
This is in agreement with the result of the standard quantum illumination protocol. We notice that $\beta_{cl}^{col}\simeq N_S/N_B$ in the $N_B\gg1$ limit. The general bound in Theorem~\ref{QCB} is tight for $N_S\ll1$ and for $N_S,N_B\gg1$. By a direct comparison, we see that the optimal achievable gain with respect to the idler-free case is $6$~dB, and the maximal advantage can be achieved when $N_B\gg1$ and $N_S\ll1$.  This result can be understood by noticing that in the low signal-photon regime the protocol can be approximated by Gaussian quantum illumination, where the detection problem can be mapped to discriminating between two coherent states with amplitude proportional to the two-mode correlations at the receiver level~\cite{Zhuang17}. Here, the Kennedy receiver is optimal, and it shows a $6$~dB advantage with respect to the idler-free case. No advantage can be detected in a vacuum environment ($N_B\ll1$), which is the case of the optical systems. Interestingly, our bounds show that even in the $N_B\gg1$ limit the gain decreases with increasing number of signal photons $N_S$, achieving the same performance of a coherent state transmitter in the $N_S\gg1$ limit. This also means that amplifying the source is not a possible to implement a quantum illumination protocol, ruling out any gain claimed in a recent experiment~\cite{Shabir}. The setup in Fig.~\ref{Setup} is particularly relevant for studying entanglement-assisted low-frequency communication in very noisy environment. However, one must note that the advantage of using quantum correlations is kept when the environment is not bright, see Fig.~\hyperref[fig:3]{3a}. For instance, a $4.6$~dB maximal advantage can be achieved for $N_B\simeq1$, implying that the advantage in using quantum correlations is not limited to the $N_B\gg1$ case.\\

\subsection{B. Quantum estimation strategies}
An approach based on the quantum estimation of the amplitude modulation has been developed in Ref.~\cite{Sanz17} in the quantum illumination context. This is less experimentally demanding than the collective strategy, as it does not require the interaction between the $M$ copies of the received signal. However, it comes at some loss in the error exponent of the EP, quantified as at least $3$~dB with respect to the optimal collective strategy. Here, we use the same concept in order to deal with the BPSK case. We address this approach as {\it local} strategy, as opposite of the collective strategies previously discussed. First, we notice that the received modes $\{\hat a_R^{(k)}\}$ can be expressed as
\begin{align}
    \hat a_R^{(k)}=\eta~e^{-i\phi} \hat a_S^{(k)}+\sqrt{1-\eta^2}~\hat h^{(k)},
\end{align}
where $\hat h^{(k)}\equiv \sqrt{\frac{\eta}{1+\eta}}~e^{-i
\phi}\hat h_\leftarrow^{(k)}+\sqrt{\frac{1}{1+\eta}}~\hat h_\rightarrow^{(k)}$ are thermal modes with $N_B$ average number of photons. We can thus optimally estimate the parameter $\kappa\equiv\eta~e^{-i\phi}\in \mathbb{R}$, obtaining a value $\kappa_{est}$, and deciding towards the hypothesis $[\phi=0]$ if $\kappa_{est}> 0$ or the hypothesis $[\phi=\pi]$ if $\kappa_{est}<0$. We refer to this strategy as ``threshold discrimination strategy".   The main figure of merit quantifying the quantum estimation performance is the quantum Fisher information (QFI), defined as~\cite{Paris09}
\begin{align}
    F=\sum_{mn}\frac{|\langle \phi_m|d\rho|\phi_n\rangle|^2}{\lambda_m+\lambda_n},
\end{align}
where $d\rho=(\partial_\kappa \rho_{\kappa})_{|\kappa=0}$, with $\rho_\kappa\equiv \rho_{\eta.\phi}$, and $\lambda_m$ is the eigenvalue of $\rho_0$ corresponding to the eigenstate $|\phi_m\rangle$. This is due to the Cramer-Rao bound, which asserts the limit of the achievable precision of an unbiased estimator $\hat \kappa$: $\Delta \hat \kappa^2 \geq 1/M F$. An estimator saturating the Cramer-Rao bound is given by the mean over the $M$ single-copy measurements of the observable $\hat O= \hat L/F$, where $\hat O =  \sum_{mn}\frac{\langle \phi_m|d\rho|\phi_n\rangle}{\lambda_m+\lambda_n}|\phi_m\rangle\langle \phi_n| $ is the symmetric logarithmic derivative computed at $\kappa=0$.  Due to the central limit theorem, the EP for the threshold discrimination strategy is $p_{\rm err}\simeq 1-\text{erf}~(\eta\sqrt{FM/2})\leq\frac{1}{2} e^{-\eta^2FM/2}$ for $M\gg1$. The previous discussion holds whenever one has an a priori knowledge of the neighborhood where $\kappa$ belongs to (in our case $\kappa\ll1$). If no assumptions of this sort can be made, generally the optimal strategy consists of a two-stage adaptive protocol: use $M^{1/\delta}$ (with $\delta>1$) copies to estimate the neighborhood and then use the rest of the copies to optimal estimate the parameter. This provides the same asymptotic performance as when the information on the neighborhood is provided. The same adaptive protocol can be used to generalize the ideas of this article to more complex alphabets, by first having a rough estimation of the phase $\phi$, and then rotate the system in order to maximize the classical Fisher information~\cite{Shi20}. 

Similarly to the case of collective strategies, we will adopt the exponent of the EP to the first relevant order in $\eta$ as figure of merit, i.e. $\beta^{loc}\equiv F/2$. Let us denote with $\beta^{loc}_{\rm free}$ the value of $\beta^{loc}$ in the idler-free case. We have the following bounds on the achievable EP decaying rate using quantum estimation methods.
\begin{theo}{\bf [Ultimate QFI bound]}\label{QFIbound}\\ 
The following bounds holds for the receiver performance with local strategies:
\begin{align}
    \beta^{loc}&\leq \min\left\{\frac{2N_S}{1+N_B},\frac{1}{\sqrt{c_B}}\frac{N_S+\frac{1}{2}}{1+N_B}\right\},\\
   \beta_{\rm free}^{loc}&\leq\frac{2}{1+c_B}\frac{N_S}{1+N_B}\equiv \beta^{loc}_{cl},
\end{align}
where $c_B=\frac{N_B}{1+N_B}$. The idler-free case is saturated by coherent state signals, in which case the optimal detector is homodyne.
\end{theo}
Similarly to Theorem~\ref{QCB}, the bound in Theorem~\ref{QFIbound} is tight for $N_S\ll1$ and for $N_S,N_B\gg1$. It follows that the maximal advantage with respect to the classical case is $3$~dB in the EP exponent if we adopt a threshold discrimination strategy. Theorem~\ref{QFIbound} also implies that the optimal detector in the classical case can be implemented with local measurements in the $N_B\gtrsim 2$ regime, as in this case $\beta_{cl}^{loc}\simeq \beta_{cl}^{col}$. However, this is not anymore valid in the $N_B\lesssim 1 $ regime, where collective measurements start to perform better, achieving $\beta_{cl}^{col}\simeq 2 \beta_{cl}^{loc}$ in the $N_B\ll1$ limit. 

\section{IV. Transmitter examples}\label{Sec:add}

The aim of this section is to discuss two topical transmitter examples whose optimal receiver saturates the ultimate bounds in the correlated case. For these states, we have derived explicit formulas of $\beta^{col}$ and $\beta^{loc}$, showing how they scale with respect to the system parameters. We also show a time-bandwidth estimation for a typical communication scenario.

\subsection{A. Two-mode squeezed vacuum state transmitter}

TMSV states are defined as
\begin{equation}\label{TMSS1}
|\psi\rangle_{{\rm TMSV}}=\sum_{n=0}^\infty \sqrt{\frac{N_S^n}{(1+N_S)^{1+n}}}|n\rangle_I|n\rangle_S,
\end{equation}
where $\{|n\rangle\}_{n=0}^\infty$ is the Fock basis. They have been thoroughly studied in the context of QI because they are a good benchmark to show a quantum advantage and because they are experimentally easy to generate, regardless of the frequency regime. The performance for the optimal receiver of a TMSV state transmitter are given by
\begin{align}
        \beta_{\rm TMSV}^{col}&=\frac{4}{\left(1+\sqrt{c_Sc_B}\right)^2}\frac{N_S}{1+N_B}\\
        \beta_{\rm TMSV}^{loc}&= \frac{2}{1+c_Sc_B}\frac{N_S}{1+N_B},
\end{align}
where $c_S=\frac{N_S}{1+N_S}$ and $c_B=\frac{N_B}{1+N_B}$ (see Appendix~\hyperref[suppl:I]{A}). The optimal receiver for a threshold discrimination strategy consists in measuring in the eigenbasis of $\hat a_I\hat a_R +\hat a_I^\dag \hat a_R^\dag$.
It is clear that TMSV states saturate the bound of Theorem~\ref{QCB} in the $N_S\ll1$ limit. In addition, the advantage of using the optimal detector for TMSV states decays slowly with increasing $N_S$, making the protocol useful also for finite number of signal photons. The detectors achieving the maximal gain are known for both the collective (only for $N_S\ll1$~\cite{Zhuang17}) and the local (for any $N_S$~\cite{Guha09}) cases. Generally, they involve efficient photon-counting devices, which are yet to be developed in the microwave regime. 
A possible solution is the use optomechanical transducers into optical frequencies, where sensitive detectors are available~\cite{Barzanjeh15}. However, current optomechanical transducers are still in infancy, as they suffer from low efficiencies and high thermal added noise.  
A different solution is to use a qubit as single-photon detector. This approach has the advantage of seamless integration with cQED platforms, as dispersive qubit measurements can be implemented already in the lab~\cite{Walter17,Buisson20}. However, single-photon detection devices so far have achieved only a $70~\%$ fidelity~\cite{Huard2020}, making this approach currently unsuitable for practical applications.
On a different note, the idler storage must be carefully considered in order to understand how quantum correlations can be useful in practice. In cQED, memory elements based on a coaxial $\lambda/4$ resonator with coherence time of nearly 1 ms have been demonstrated~\cite{Reagor16}. This corresponds to $300$~Km of free-space propagation of light. Another promising alternative consists in transferring the idler bosonic degrees of freedom to the delay line based on surface acoustic waves~\cite{SAW1,SAW2}. 

\subsection{B. Schr\"odinger's cat state transmitter}
We now discuss the performance of the protocol based on SC states, created by the interaction of a qubit with a continuous-variable signal. It comes as no surprise that SC states show the same advantage of Gaussian states for $N_S\ll1$, because in this limit the two states approximate each other. However, the underlying physics is different. In fact, the SC scheme relies on qubit measurements, and provides us naturally with a digital way to store the idler in a cQED setup.
 
\begin{equation}\label{SCS}
    |\psi\rangle_{\rm SC} = \frac{1}{\sqrt{2}}\left[|+\rangle_I | \alpha\rangle_S + |-\rangle_I |-\alpha\rangle_S\right],
\end{equation}
where $|\pm\rangle=\frac{1}{\sqrt{2}}(|g\rangle \pm |e\rangle)$ are eigenstates of the Pauli operator $\hat \sigma_x$, $|\alpha\rangle$ is a coherent state with amplitude $\alpha>0$ ($N_S=|\alpha|^2$), assumed to be real for simplicity. Of particular interest will be the case of $|\alpha|\ll1$, that shows the maximal advantage with respect to the classical case. This state can be written in the Schmidt decomposition as 
\begin{equation}\label{SCShmidt}
    |\psi\rangle_{\rm SC}= \sqrt{\lambda_+}|g\rangle|\alpha_+\rangle + \sqrt{\lambda_-}|e\rangle|\alpha_-\rangle,
\end{equation}
where $\lambda_{\pm}=\frac{1}{2}(1\pm e^{-2N_S})$ and $|\alpha_\pm\rangle=\frac{1}{2\sqrt{\lambda_{\pm}}}[|\alpha\rangle\pm|-\alpha\rangle]$.
The performance of the optimal receiver for a SC-state transmitter are given by 
\begin{align}
    \beta^{col}_{\rm SC}&= f^{col}_{\rm SC}(N_S,N_B)\frac{N_S}{1+N_B},\\
    \beta^{loc}_{\rm SC}&= f^{loc}_{\rm SC}(N_S,N_B)\frac{N_S}{1+N_B},
\end{align}
where 
\begin{align}
        f^{col}_{\rm SC} &\stackrel{N_B\gg1}{=} 4-8\sqrt{N_S} +O(N_S)\\ 
        f^{loc}_{\rm SC} &\stackrel{N_B\gg1}{=}2-4N_S+O(N_S^2). 
\end{align}
The optimal threshold discrimination strategy in the $N_B\gg1$ regime consists in measuring in the eigenbasis of the observable $\hat O_{opt}=\hat \sigma^-[\lambda_+ \hat a_R +\lambda_-\hat a_R^\dag] +c.c.$.
The exact expressions of $f^{col}_{\rm SC}$ and $f^{loc}_{\rm SC}$ are given in Appendix~\hyperref[suppl:I]{A}. A comparison with the TMSV state case is shown in Fig.~\hyperref[fig:3]{3b}. As expected, the maximal gain can be achieved for $N_S\ll1$. In addition, the gain decays exponentially with increasing $N_S$. In fact, the observable $\hat \sigma_x(\hat a_R+\hat a_R^\dag)$ is optimal for $N_S\gtrsim 1$, therefore the classical mixed state $\frac{1}{2}[|+\rangle_{I}\langle+|\otimes|\alpha\rangle_{S}\langle \alpha|+|-\rangle_{I}\langle-|\otimes|-\alpha\rangle_{S}\langle -\alpha|]$ performs the same as $|\psi\rangle_{\rm SC}$ in this regime. This loss of gain for finite $N_S$ can be mitigated by considering states with larger Schmidt rank, i.e. $\frac{1}{\sqrt{d}}\sum_{k=0}^{d-1} |w_{k}\rangle_I|\alpha_k\rangle_S$, where $\alpha_k=\sqrt{N_S}e^{i\frac{2\pi k}{d}}$ and the idler is a $d$-level system with $\langle w_k|w_{k'}\rangle=\delta_{k,k'}$~\cite{Sanz17}. This state can be implemented by letting several transmon qubits interacting with the same resonator. However, in this case the optimal detector is complicated and we will not consider it in the implementation discussion. 

\subsection{C. Time-bandwidth estimation}

If we consider a microwave implementation, where  $N_B\simeq10^3$ for $T=300$~K at $5$~GHz, we need a time-bandwidth product $M\simeq \frac{\beta_{cl}}{\beta_{\rm TMSV}}\frac{N_B}{\eta^2 N_S}\log(1/p_{err})$ to reach an error probability of $p_{err}$. If we assume propagation in free space at ambient condition, where power losses are around $l\simeq0.01$~dB/Km for the considered spectrum, and an antenna area of $0.1$~m$^2$ for $R\simeq1$~m space propagation, this gives us $\eta=10^{-Rl/10}\frac{A_R}{4\pi R^2} \simeq 10^{-2}$~\cite{Karsa2020, Dic}.  If we choose $N_S\simeq 0.01$ average number of photons per mode, where TMSV states show a close to optimal gain with respect to the classical protocol, then we require a time-bandwidth product of $M\simeq 10^9 \frac{\beta_{cl}}{\beta_{\rm TMSV}} \log(1/p_{err})$. The error probability is chosen depending on the circumstance. In a communication scenario, one needs to perform error correction where in practice error probabilities of the order of $10^{-2}$ are needed. This gives us a time bandwidth product of the order of $M\sim 10^{9}$. This number can be considerably lowered by using directional antennas. These are made typically as phased arrays, comprising tens or hundreds of small elements with a phase difference between them. This phase difference is chosen such that the radiated field is maximum, due to constructive interference, along the desired direction.~\cite{multi} Assuming that an effective way of measuring the power of large bandwidth signals with few-photons sensitivity is available, then a Gaussian state protocol is arguably the best option for implementing the ideas presented this article. In fact, the ability of generating large bandwidth Gaussian signals would reduce the time complexity of the protocol. Even though at the present stage Gaussian states remains the main solution for sensing and metrology in noisy regime, SC states provides a way to avoid photon-detection. 

\section{V. Covert quantum communication}\label{sec:IV}
We now discuss the possibility of performing secure quantum communication using the setup depicted in Fig.~\ref{Setup}, by exploiting recent results about covert quantum communication~\cite{Bash15,Arrazola16,Bullock2020}. The basic idea is to protect the content of the message to be transmitted by covering the existence of the carrier in a given bandwidth and temporal frame. In this context, Eve's main task becomes to understand whether a message is being transmitted or not through the channel. Here, we provide bounds for Eve's detection probability depending on the receiver performance. The results of this section holds for states which well approximates Gaussian states in the $N_S\ll1$ limit. We prove that $\bar m=O(\sqrt{n})$ number of bits securely transmittable over $n$ channel usages with an arbitrary small EP. We also show how quantum correlations can increase the bit transmission rate by a constant factor, which depends on the adopted strategy (i.e. collective or local). 

\subsection{A. Square-root law}
A natural way of defining covertness consists in bounding from below the probability of detecting that communication between Alice and Bob is happening.
\begin{definition}{\bf [Covertness criteria]} A communicating system is $\delta$-covert over $n$ channel usages if Eve's EP in discerning between the equally likely hypothesis of communication happening or not-happening is $P^{{\rm (Eve)}}\geq \frac{1}{2}-\delta$ for $n$ large enough.
\end{definition}
Ideally, we would like to have $\delta$ as smaller as possible by still being able to communicate a finite number of bits. 
Generally, covert communication is possible because Eve does not have control at least to a part of the environmental channel~\cite{Bash15}. This assumption is not radical, as in the low-frequency regime at room temperature there is an unavoidable noise dictated by the laws of physics. We assume that Bob's and Alice's places, where the state manipulation and the measurements are implemented, are sealed, and that the signals are sent directly to a room temperature environment where Eve may be placed. We also assume, for simplicity, that the part of the channel that Eve cannot control does not change while communication is in progress. The latter assumption can be relaxed by analyzing more general fading communication channel models, where the amplitude losses and/or the signal phases are random variables~\cite{Zhuang2017}. In addition, 
we provide Alice with the capability of implementing {\it truly} random phase modulations on her signal on the alphabet $\mathcal{A}$. This is an important requirement for ensuring covertness in the two-way setup. 

Alice and Bob use $n=mM$ modes of the bosonic channel simultaneously. They use $M$ modes to transmit a symbol taken from a discrete alphabet $\mathcal{A}$.  In addition, they use a publicly available codebook $\mathcal{C}$ that maps $\bar m$-bit input blocks to $m$-symbol codewords from $\mathcal{A}^m$, with $\bar m<m$, by generating $2^{\bar m}$ codeword sequences, i.e. $\mathcal{C}=\{a_k\in\mathcal{A}^m\}_{k=1}^{2^{\bar m}}$.  The codebook is built in the way that the codewords, when the transmission is corrupted by the channel, are distinguishable from each other with high probability. This induces a natural way of defining when communication is reliable.
\begin{definition}{\bf[Reliability criteria]} A communicating protocol is $\epsilon$-reliable if the decoding error probability averaged over the codebook is bounded by $\epsilon$, i.e. when $\frac{1}{|\mathcal{C}|}\sum_{a_k\in\mathcal{C}}\sum_{a_j\in\mathcal{C}\setminus\{a_k\}}P(a_j|a_k)\leq \epsilon$ for $n$ large enough.
\end{definition}
We focus primarily to the BPSK case,  corresponding to $\mathcal{A}=\{0,\pi\}$, keeping in mind that the concept can be generalised to more complex constellations.  Each symbol transmission is done by performing the two-way protocol described in Fig.~\ref{Setup}. We define the $on$-setting, corresponding to the case when the communication is happening, and the ${\it off}$-setting, when no information is exchanged between Alice and Bob. In other words, we consider the $on$-setting when Alice and Bob applies the protocol with $N_S>0$, while the ${\it off}$-setting is when $N_S=0$. We consider a passive eavesdropper, able to catch all the modes that are lost in the Bob-Alice path, denoted with the $\leftarrow$ subscript, and Alice-Bob path, denoted with the $\rightarrow$ subscript.  In the following we will consider the worst-case scenario, where Eve gets all the lost photons in the environment, since the case of Eve with limited capabilities can be similarly derived. Using directional antennas would decrease the photons received by Eve in the worst-case scenario, improving the covert bit transmission rate. For a given slot, Eve gets the modes
\begin{align}
   \hat  w_{\leftarrow}^{(k)}&= -\sqrt{1-\eta}\; \hat a_{S}^{(k)}+ \sqrt{\eta}\; \hat h_\leftarrow^{(k)} \label{on1}\\
    \hat w_{\rightarrow}^{(k)}&= -\sqrt{1-\eta}\; e^{-i\tilde \varphi_k}\hat a'^{(k)}_S+\sqrt{\eta}\; \hat h_\rightarrow^{(k)}, \label{on2}
\end{align} 
for $k=1,\dots, M$. Here, $\{N_S>0,\,\tilde \varphi_k=\varphi_k+\phi\}$ defines the $on$-setting, while $\{N_S=0,\,\tilde \varphi_k=\varphi_k\}$ is the $\it off$-setting. The goal is to let the $on$ and the ${\it off}$ settings the least distinguishable possible. This is possible only in the $N_S\ll1$ limit, as in this case Eve's mode in both settings approximate each other.  This is due to the fact that both $\varphi_k+\phi$ and $\varphi_k$ are distributed uniformly at random in the alphabet $\mathcal{A}$. The inclusion of the random sequence of phase-shifts by Alice is a crucial requirement for the covertness proof, as otherwise Eve would have enough resources to uncover the communication by detecting the phase $\phi$ in a given slot, by let interfere the modes $\hat  w_{\leftarrow}^{(k)}$ and $\hat  w_{\rightarrow}^{(k)}$. However, she can still detect if communication is happening by detecting the changes in power of each path, and their correlations. As all the $\hat a_S^{(k)}$ are i.i.d., Eve's quantum state does not depend on $k$ and on which symbol is being transmitted.

\begin{lemma}{\bf[Covertness: achievability and converse bounds]}\label{theocovert}
Let $N_S>0$ be the average number of signal photons in the $on$-setting. Let the signal density matrix be $\rho_S=\sum_{j=0}^\infty N_S^j\sigma_j$, where 
\begin{align}
    \sigma_0&= |0\rangle\langle0|\\
    \sigma_1&= |1\rangle\langle1|-|0\rangle\langle0| \\
    \sigma_2&= c(|2\rangle\langle2|-2|1\rangle\langle1|+|0\rangle\langle0|),
\end{align}
with $0\leq c\leq1$. Then, the communication between Alice and Bob is $\delta$-covert over $n$ channel usages, provided that  $N_S\leq A_{\eta,N_B}\frac{\delta}{\sqrt{n}}$, with $A_{\eta,N_B}=\frac{4\sqrt{\eta^2N_B(1+\eta^2N_B)}}{(1-\eta^2)}$. In addition, no constant better than $A_{\eta,N_B}$ can be found. 
\end{lemma}
Here, we learn that the correlations between the two paths are not useful for Eve to detect the presence of the transmitted signal, provided that random modulations of the signal are applied at Alice. Notice that the constant $A_{\eta,N_B}$ is the same as in the one-way covert protocol with the change $\eta\rightarrow \eta^2$~\cite{Bash15}, which is due to the double path transmission.
Lemma~\ref{theocovert} can be directly applied to a TMSV state transmitter, which corresponds to $c=1$. It can be also applied to coherent state and SC state transmitters, if we allow Bob to perform random phase modulations. In fact, let $|\alpha_k\rangle$ with $\alpha = |\alpha|e^{-ik \pi /4}$ be a coherent state, then $\rho_S=\frac{1}{8}\sum_{k=0}^7|\alpha_k\rangle\langle \alpha_k|$ respects the conditions of Lemma~\ref{theocovert} with $c=1/2$. Bob's phase modulation at the transmission can be reversed at the receiver level due to the linearity of the communicating channel. We also notice that a Gaussian thermal state at Bob's side is not needed in order to ensure covertness, meaning that complex Gaussian modulations of Bob's signal are not needed. We can rely instead on discrete phase modulations, which are experimentally easier to generate and they require less memory complexity.

We have provided an upper bound on the average transmitting power $N_S$, which needs to scale as $O(\sqrt{n})^{-1}$ in order to keep the communication covert over $n$ channel usages. Typical transmitter operates at constant photon number per mode, and the requirement of $N_S$ decaying with the inverse of $\sqrt{n}$ can be quite restrictive. This constraint can be relaxed by defining a {\it probabilistic} version of the protocol, which makes use only of a fraction $\tau\leq A_{\eta,N_B}\frac{\delta}{N_S\sqrt{n}}$ of the $n$ available modes in the $on$-setting~\cite{Bullock2020}, see Fig.~\ref{covfig}. In each of these modes, the transmitting power $N_S$ is kept constant and small.   
\begin{figure}[t!]
\includegraphics[width=0.38\textwidth]{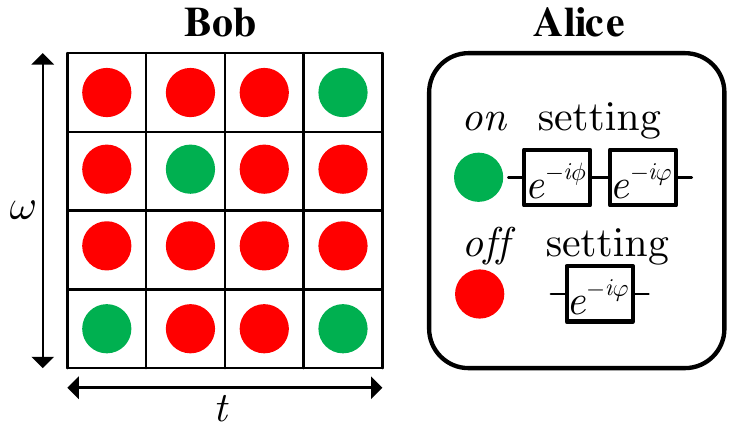}
\caption{{\bf Probabilistic version of the two-way covert quantum communication protocol.} Alice and Bob secretly choose to communicate only in the fraction of the temporal and frequency modes at their disposal. In the $on$-setting (green dots), Alice performs a phase modulation $e^{-i(\phi+\varphi)}$, where $\phi\in\mathcal{A}$ is the symbol that Alice wants to transmit and $\varphi$ is taken uniformly at random in $\mathcal{A}$. The phase $\varphi$ is secretly known to both Alice and Bob. In the ${\it off}$-setting (red dots), Alice performs a phase modulation $e^{-i\varphi}$, where $\varphi$ is taken uniformly at random in $\mathcal{A}$ Here, $\phi$ is generated in situ and it is not shared by Alice.}
\label{covfig}
\end{figure} 

Lemma~\ref{theocovert}, together with a random-coding argument, implies that the square-root-law is achievable by our two-way setup.

\begin{theo}{\bf [Square-root law: two-way setup]}\label{SRL}
Let Alice and Bob share a publicly available codebook and a secret random sequence of length $n$ [$O(\sqrt{n}\log n)$ in the probabilistic version]. Then, they can communicate $\bar m=\beta_{det} \beta_{cov}\delta \sqrt{n}+\log_2\epsilon$ bits over $n$ channel usages by keeping the system $\delta$-covert and the communication $\epsilon$-reliable. Here, $\beta_{cov}=\frac{8}{\pi\log 2}c_B 
\eta^4$ provided that the transmitted signal satisfies the assumptions of Lemma~1, and $\beta_{det}$ is a constant that depends on the detector:  $\beta_{det}\leq4$ ($\beta_{det}\leq2$) for the TMSV state and SC state transmitters with a collective (local) receiver, and $\beta_{det}=1$ for the coherent state transmitter with a homodyne receiver. 
\end{theo} 

This result comes from a compromise between the transmission rate needed to keep the communication covert, quantified as $N_S=O(\sqrt{n})^{-1}$ in Lemma~\ref{theocovert}, and the ability of performing error correction on a random code, whose error probability decreases exponentially with the total transmitted power $nN_S=O(\sqrt{n})$. In this way, the error probability $P$ in transmitting $\bar m$ bits is $\log P\sim \bar m-O(\sqrt{n})$, which implies the square-root scaling. In Theorem~\ref{SRL}, all the constants have been explicitly derived for different transmitters and receivers for the two-way protocol.

\subsection{B. Key expansion and synchronization}
Summarizing, Alice and Bob need to agree secretly on the following information prior the communication for each $n$ channels in the {\it on}-setting: (i) A secret random sequence corresponding to the random phase-shifts by Alice. 
This information requires $O(n)$ bits of pre-shared knowledge, or $O(\sqrt{n})$ in the probabilistic version. (ii) In the probabilistic version, the information needed to specify the modes which are used in the $on$-setting. This requires $O(\sqrt{n}\log n)$ bits of shared secret. Moreover, Alice and Bob need to agree on which temporal frame to turn the communication on, which requires additional $O(N)$, where $N=\frac{\Delta T\Delta \Omega}{n}$ is the number of temporal frames in a time interval $\Delta T$ assuming an operational bandwidth $\Delta \Omega$. Assuming $N=C(n)$, with $C(n)>n$, eventually we will need a key of at least $\omega\left(\frac{C(n)}{\sqrt{n}}\right)$ bits of pre-shared secret to transmit $O\left(\frac{C(n)}{\sqrt{n}}\right)$ covert bits. It is then clear that there is an overhead of number of pre-shared bits with respect to the transmitted ones, if we want to ensure covertness. Therefore, Alice and Bob need to meet regularly in order to agree on a key. These meetings can be implemented via public communication through an authenticated channel. In this case, the key generation and expansion steps can be done covertly and secretly with a protocol based on likelihood encoder techniques, assuming that the public channel is classical~\cite{Bloch19, Song}. Once Alice and Bob share the key, we have shown that using the Vernam cypher allows for covert and unconditionally secure communication. 
If instead the parties are not allowed to publicly communicate at any point, but they still do pre-share a secret key, one can expand the key using a pseudo-random generating function. Indeed, if Alice and Bob pre-share a pseudo-random generating function $f:\{0,1\}^l\rightarrow \{0,1\}^{p(l)}$, where $l\in\mathbb{N}$ and $l=o(p(l))$ for $l\gg1$, then it is possible to communicate more bits than the pre-shared ones.  Widely used Advanced Encryption Standard and the Secure Hashing Algorithm have outputs that are exponentially larger than their seeds while still retaining computational indistinguishability from true randomness~\cite{Arrazola18}. The security of the protocol defined by a key expansion step based on pseudo-random functions followed by a one-time-pad encoding is not unconditionally secure. Its security relies on the computational assumption that hacking the pseudo-random function is hard in a way. However, covertness adds a new layer of security, since the transmission is also protected from the limited Eve's detection capability and the bandwidth spreading. Indeed, let us consider the probabilistic version of the protocol, where only a small fraction $\tau$ of the modes is used to communicate. Since Alice applies to the remaining modes truly random phase modulations generated in situ, the randomness of the transmitted string can be controlled with $\tau$. One may think a similar protocol in the one-way case, by generating meaningless signals in the {\it off}-setting. However, generally the covertness criteria is whether a signal is being transmitted or
not rather than if Alice and Bob are transmitting anything meaningful or not. This is because there
are relevant scenario when the transmitter must indeed be switched off~\cite{Wang2016}.  Our two-way protocol, instead, uses passive operations to embed the information, which allows to generate meaningless messages when the transmitter is switched off.

To establish covert communication, Alice and Bob require a pre-established secret, in
this case the knowledge of the phases $\varphi_k$ in the $on$-setting and the instances when communication is happening. This in turn requires them to have very accurate
synchronization so that time tags of signals are properly assigned. The simplest solution for synchronization is that Alice and Bob are both synchronized to the same clock. Sub nano-second level clock accuracy is possible using IEEE 1588 high accuracy default profile~\cite{Girela2020}, but that requires physical cabling between the master clock and the communicating party. High accuracy wireless synchronization is possible using General Global Navigate Satellite System (GNSS)~\cite{Weinbach2011}. Since our scheme requires continuous synchronization between Alice and Bob in both the cases of communication switched on and off, this step does not need to be implemented covertly.

\section{VI. circuit QED Implementation}\label{sec:V}

The ideas introduced in this paper can be easily implemented using coherent states properly modulated, as described in the previous section, and heterodyne measurement. While this scheme does not achieve any of the ultimate bounds for the receiver, it is the most practical way of realizing the protocol. If one is allowed to implement certain entangling operations at the transmitting and receiver level, then larger key rates are achievable. In this context, optimal schemes for the Gaussian state receiver have been thoroughly studied in the literature. Instead, a receiver for the SC state transmitter is still missing. In this section, we fill this gap by introducing an implementation in a circuit QED setup of a transmitter and a receiver based on SC states. We show that Jaynes-Cumming (JC) operations and qubit measurements are enough to fully implement the protocol. This comes with a great advantage with respect to the Gaussian state receiver, which requires photo-detectors. We also provide an analysis of how the decoherence affects the protocol based on quantum correlations. The discussion will be mostly at the model level. However, it is noteworthy to observe that all the operations described in the following have been proved in cQED since fifteen years, with increasing enhancements of fidelity for the gate implementation and state storage.

\subsection{A. State preparation}

The SC state defined in Eq.~\eqref{SCS} can be prepared in a circuit QED setup as described in the following. Consider the JC Hamiltonian
\begin{equation}\label{JCHam}
\hat H=\hat H_0^{(\omega_r,\omega_q)} + \hat H_{JC}^g,
\end{equation}
where $\hat H_0^{(\omega_r,\omega_q)}=\hbar \omega_r \hat a^\dag \hat a+\frac{\hbar\omega_q}{2}\hat \sigma_z$, with $\hat \sigma_z=|e\rangle\langle e|-|g\rangle \langle g|$, and $\hat H_{JC}^g=\hbar g (\hat \sigma^+ \hat a +\hat \sigma^- \hat a^\dag)$.
Here, $\omega_r$ and $\omega_q$ are the frequency of the resonator and the qubit respectively, and $g$ is the coupling between these two systems. 
We also define the  detuning  $\Delta = \omega_r-\omega_q$ and 
$\Gamma=\max\{\kappa, 1/T_1, 1/T_2\}$, where $\kappa$ is the cavity decay rate, and $T_1$ and $T_2$ are the qubit decaying and dephasing times respectively (see Appendix~\hyperref[suppl:IV]{D}). In the dispersive regime, where $\Delta\gg g$, one can apply perturbation theory to the first order of the parameter $g/\Delta$, finding the effective Hamiltonian
\begin{equation}
    \hat H_{SDR}= \hat H_0^{(\omega_r,\omega_q+\chi)} + \hbar \chi \hat \sigma_z \hat a^\dag \hat a,
    \end{equation}
where $\chi=g^2/\Delta$~\cite{Gambetta}. We assume that $\chi\gg \Gamma$, which is known as the {\it strong-dispersive regime} (SDR). In this way, any losses of the bosonic mode and the qubit are negligible during the implementation of the gate, as long as the operating time will be sufficiently short. The preparation protocol is based on the fact that the Hamiltonian $\hat H_{SDR}$ is a conditional phase-shift on the resonator, with the qubit acting as the control. To put this easier in evidence, we will work in a rotating frame defined by the free Hamiltonian $\hat H_0^{(\omega_r,\omega_q+\chi)}$. The Hamiltonian in the rotating frame is $\chi \sigma_z\hat a^\dag \hat a$. The preparation protocol consists in the following steps~\cite{Vlastakis}, assuming an initial qubit-resonator state $|g\rangle|0\rangle$ (see Fig.~\ref{draw5}).
\begin{description}
\item[{\bf Step 1}] Apply a $\pi/2$ $\hat \sigma_y$-pulse to the qubit in the ground state, and drive the resonator at frequency $\omega_r$ with a signal calibrated such that the coherent state $|-i\alpha\rangle$ is prepared.

\item[{\bf Step 2}] Let the qubit and the resonator interact for a time $t_\chi=\frac{\pi}{2\chi}$. This results in a conditional phase shift on the cavity state by the operator $|g\rangle \langle g|\otimes \exp(i \pi \hat a^{\dag} \hat a /2) + |e\rangle \langle e|\otimes \exp(-i \pi \hat a^{\dag} \hat a /2) $, Its action on the state prepared in Step 1 can be understood as a uniform counterclockwise rotation by an angle $\pi /2$ of the coherent state, followed by the application of the photon parity operator $\exp(-i \pi \hat a^{\dag} \hat a )$ if the qubit is excited. The state after this step is $\frac{1}{\sqrt{2}}\left[{|g\rangle|\alpha\rangle + |e\rangle| -\alpha\rangle}\right]$.

\item[{\bf Step 3}] Apply a $\pi/2$ $\hat \sigma_y$-pulse to the qubit. The state after this step is $\frac{1}{\sqrt{2}}\left[{|+\rangle|\alpha\rangle + |-\rangle| -\alpha\rangle}\right]$.
\end{description}

In order to implement Step 1 and Step 3 we would need to decouple the qubit and the resonator: this can be achieved either by a tunable coupler or by further detuning the qubit. Also we have considered here the ideal situation when all the operations can be realized with high fidelity, which is a good approximation in the strong regime~\cite{Wang16}. The main remaining source of errors is due to the spurious thermal contribution present in the cryogenic environment prior to the preparation stage. This will be the object of a later on discussion.

\begin{figure}[t!]
\includegraphics[width=0.30\textwidth]{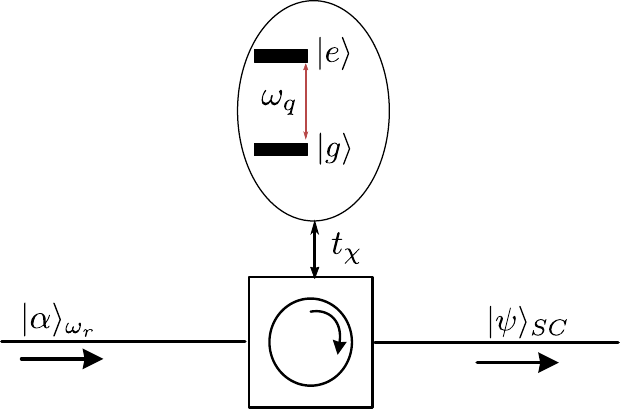}
\caption{Scheme for the preparation of the Schr\"odinger's cat state. A resonator with central frequency $\omega_r$ is driven with a coherent signal, displacing the state of the cavity to $|-i\alpha\rangle$. A transmon qubit with frequency $\omega_q$ is initialized in a state $|+\rangle=\frac{|e\rangle+|g\rangle}{\sqrt{2}}$. A conditional phase shift is then applied. This is implemented by letting the resonator and qubit interact in the strong dispersive regime for a time $t_\chi=\pi/2\chi$, where $\chi=g^2/\Delta$ is the effective coupling. Here, $g$ is the coupling strength of the qubit-resonator system, $\Delta=\omega_r-\omega_q$ and $g/\Delta\ll1$. Finally, a $\pi/2$ $\hat \sigma_y$-pulse is applied to the qubit. Feasible parameters are $\omega_r=5$~GHz, $\omega_q=\omega_r+\Delta$ with $\Delta=20$~Mhz, and $g=100$~Khz.}
\label{draw5}
\end{figure} 

\subsection{B. Receiver for the entanglement-assisted protocol}

For the implementation of the optimal observable $\hat O_{opt}$ we will make use of the JC Hamiltonian defined in Eq.~\eqref{JCHam} in the {\it strong-resonant} regime, i.e. when $\omega_q=\omega_r$ and $g\gg \Gamma$. The qubit-resonator system evolution under a time $t_g=\tau/ g$ corresponds to applying the gate $\hat U_\tau=e^{-\tau [\hat a_R^\dag\hat \sigma^- - \hat a_R \hat \sigma^+]}$ up to a known phase shift $e^{-i \hat H_0t_g/\hbar}$, as $[\hat H_0^{(\omega_q,\omega_q)},\hat H_{JC}^g]=0$. The observable $\hat O_{opt}$ can be implemented in an approximately in the following way, see Fig.~\ref{draw6}.

\begin{description}
    \item[{\bf Step 1}] Perform a squeezing operation $\hat S(r)$ on the reflected mode $\hat a_R$, with squeezing parameter $r=-\arcsinh \lambda_-$~\cite{footnote}. This generates the mode $\hat a'_R=\lambda_+ \hat a_R +\lambda_-\hat a_R^\dag$.
    \item[{\bf Step 2}] Apply a $\pi/2$ $\hat \sigma_x$-pulse to the qubit state. This switches $\hat \sigma^-$ with $
    \hat \sigma^+$.
    \item[{\bf Step 3}] Let the qubit-signal system interact with the JC Hamiltonian in the strong-resonant regime for a time $t_g=\tau/g$, with a small enough $\tau$. This generates the transformation
    \begin{align}\label{transf}
      \hat V^\dag |e\rangle\langle e| \hat V= |g\rangle\langle g| + \tau \hat O_{opt}+o(\tau),
    \end{align}
   where $\hat V= \hat U_\tau \hat S(r)\hat \sigma_x$.\\ 
    \item[{\bf Step 4}] Measure the qubit in the basis $\{|g\rangle\langle g|, |e\rangle\langle e| \}$. 
\end{description}

If  $\tau$ is low enough, this protocol approximates the measurement in the low-energy eigenspace of $\hat O_{opt}=\hat \sigma^-[\lambda_+ \hat a_R +\lambda_-\hat a_R^\dag] +c.c.$, which, in $N_S\ll1$ regime, is the relevant part of the Hilbert space. Let us define $\hat O_\tau=\hat V|e\rangle \langle e| \hat V^\dag$. The threshold discrimination protocol consists in repeating the steps 1-4 $M$ times, collecting the results $\{o_i\}_{i=1}^M$. Here, $o_i=1$ (or $0$) is the measurement outcome corresponding to the projection on the state $|g\rangle$ (or $|e\rangle$). We then calculate the relative frequency $\frac{1}{M}\sum_{i=1}^M o_i$, which corresponds to the expected value of the observable $\hat O_\tau$ on the [qubit]-[reflected signal] system state. We use the result to discriminate between the two hypothesis: $\langle \hat O_\tau \rangle_{\phi=0}=\lambda_++ \tau \eta \sqrt{N_S}(1+e^{-4N_S})+o(\tau)$ and $\langle \hat O_\tau\rangle_{\phi=\pi}= \lambda_+- \tau \eta \sqrt{N_S}(1+e^{-4N_S})+o(\tau)$.
We choose the $\tau$ value in order to maximize the SNR $Q_{\hat O_\tau}$  for the observable $\hat O_\tau$, defined as
\begin{align}
    Q_{\hat O_\tau} \equiv \frac{(\langle \hat O_\tau\rangle_{\rho_{\eta,\pi}}-\langle\hat O_\tau\rangle_{\rho_{\eta,0}})^2}{\Delta \hat O_\tau^2},
\end{align}
where $\Delta  \hat O_\tau^2  = \frac{1}{4}\left[\sqrt{\Delta  \hat O_{\tau,\phi=\pi}^2} + \sqrt{\Delta  \hat O_{\tau,\phi=0}^2} \right]^2$ is the variance of the observable $\hat O_\tau$ averaged over the states $\rho_{\eta,\phi=\pi}$ and $\rho_{\eta,\phi=0}$. The SNR is related to the EP of a threshold discrimination strategy with $p_{\rm err}\sim \exp{\left[-\frac{Q_{\hat O_\tau}M}{8}\right]}$ for $M\gg1$. In Appendix~\hyperref[suppl:III]{C}, we show that any value $N_S/N_B\ll \tau^2 \ll1/N_B$ is good for approximating the optimal SNR in the $N_B\gg1$, $N_S\ll1$ regime. For, instance, if we choose  $\tau^{2}=\frac{N_S}{\sqrt{N_B}}\equiv\tau^{*\,2}$, with $N_S\ll \frac{1}{\sqrt{N_B}}$, we obtain 
\begin{align}
\frac{Q_{\hat O_{\tau^*}}}{Q_{\hat O_{opt}}} \simeq 1-\frac{1}{\sqrt{N_B}}.
\end{align}

Interfacing a signal with $N_B\sim10^3$ number of photons with a low-temperature environment is challenging. An initial attenuation is needed, making the protocol less efficient in terms of the SNR. An attenuation can be modeled with the beamsplitter input-output relations $\hat a_{R,att}= \sqrt{\eta_{att}}\, \hat a_R+\sqrt{1-\eta_{att}}\, \hat v$, where $\eta_{att}$ is a power attenuator and $v$ is a bosonic mode assumed to be in a vacuum state. This can be achieved with cryogenic microwave attenuators. The measurement protocol is applied to the mode $\hat a_{R,att}$, resulting in a rescaled SNR: $Q^{att}_{O_{\tau^*}}/ Q_{O_{\tau^*}}\simeq \frac{\eta_{att}N_B}{1+\eta_{att}N_B}$ in the $N_S\ll1$ limit, This means that the performance is not affected as long as $\eta_{att}N_B$ is kept large enough.

\begin{figure}[t!]
\includegraphics[width=0.32\textwidth]{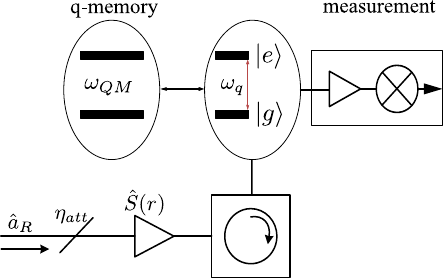}
\caption{Scheme for the implementation of the observable optimizing the quantum Fisher information. During the signal transmission, which happens at frequency $\omega_r$, the qubit state is transferred to a quantum memory. The qubit frequency is then tuned to $\omega_q=\omega_r$. Here the quantum memory is a resonator at frequency $\omega_{QM}$,  interacting dispersively with the qubit. This allows the implementation of the cat code. The state is transferred back to the qubit for the measurement stage. The received signal $\hat a_R$ is attenuated. A squeezing operation is then applied using a Josephson parametric amplifier in the degenerate mode. The output interacts with the qubit in the resonant regime. Finally, the qubit is measured in the $\{|g\rangle\langle g|,|e\rangle\langle e|\}$ basis.}
\label{draw6}
\end{figure} 

\subsection{C. Effects of decoherence on the performance}

In the state preparation scheme, the main source of inefficiency is given by the spurious thermal contribution present in the cryogenic environment prior to the preparation stage. In addition, while any sort of signal dissipation after the state preparation is already included in the quantum illumination setup in an effective way, the idler decoherence must be characterized and bounded in order to understand the actual performance of the protocol in practical scenarios. Let us first discuss how the protocol is affected by the initial thermal noise. We assume a Markovian environment at temperature $T$, whose Lindblad master equation is  $\partial_t \rho=[\mathcal{L}_D^q+\mathcal{L}_D^r] \rho$. Here,
\begin{align}\label{losqub}
    \mathcal{L}_D^q/\hbar = \frac{\gamma}{2} \mathcal{D}[\hat \sigma_z]+\Gamma_{\uparrow}\mathcal{D}[\hat \sigma^+]+\Gamma_{\downarrow}\mathcal{D}[\hat \sigma^-],
\end{align}
models the qubit decoherence, and 
\begin{align}
    \mathcal{L}_D^r/\hbar = \kappa(1+N_T)\mathcal{D}[\hat a]+\kappa N_T\mathcal{D}[\hat a^\dag],
\end{align}
with $N_T=(e^{\beta \hbar \omega_r}-1)^{-1}$ and $\beta=(k_BT)^{-1}$, is the resonator dissipation in an environment at temperature $T$. Here, the Lindblad operators act on a general qubit-resonator state $\rho$ as $\mathcal{D}[\hat L]\rho=\hat L\rho \hat L^\dag -\frac{1}{2} \{\hat L^\dag\hat L ,\rho\}$. In addition, the relations $a_g\equiv \frac{\Gamma_{\downarrow}}{\Gamma_{\downarrow}+\Gamma_{\uparrow}}=(1+e^{-\beta \hbar \omega_q})^{-1}$ and $a_e\equiv \frac{\Gamma_{\uparrow}}{\Gamma_{\downarrow}+\Gamma_{\uparrow}}=e^{-\beta \hbar \omega_q}$ hold for a qubit in an environment at temperature $T$. In a $T\simeq 20$~mK environment we have that $\beta\hbar\omega_{r,q}\gg1$ for $\omega_{r,q}\sim1-10$~GHz, therefore decoherence and dissipations in principle should not play a role in the performance evaluation. However, small thermal contributions can be relevant in the low-photons regime, and their effects on the SNR need to be quantified. We assume the initial qubit-resonator state to be the steady state of the Lindblad master equation, i.e. $\rho_q\otimes \rho_r$ with $\rho_q=a_g|g\rangle\langle g|+a_e|e\rangle\langle e|$ and $\rho_r=\frac{1}{1+N_T}\sum_{n=0}^\infty \left(\frac{N_T}{1+N_T}\right)^n|n\rangle\langle n|$. By applying the state preparation protocol, we obtain the state $\rho_{\mathrm{noisy}}= \frac{1}{2}\sum_{k,k'\in\{+,-\}}[a_g+kk'a_e]|k\rangle\langle k'| \otimes D(k\alpha)\rho_rD(k'\alpha)$, where $D(\beta)$ is a displacement operator.  This implies a rescaling of the optimal SNR for fixed transmitting power, given by $Q_{\hat O_{opt}}^{\mathrm{noisy}}/Q_{\hat O_{opt}}=\frac{(a_g-a_e)^2}{1+c^{-1}}$, where we have set $|\alpha|^2=cN_T$. We notice that for $c\leq [2(a_g-a_e)^2-1]^{-1}$ we cannot have any quantum advantage. This sets an upper limit to the amount of thermal noise tolerable before losing all the advantage with respect to a coherent-state transmitter. The initial thermal contribution can be experimentally characterized in several ways. For instance, recently a primary thermometry for propagating microwaves with sensitivity of $4\times10^{-4}$
photons/$\sqrt{{\rm Hz}}$ and a bandwidth of $40$~MHz has been developed~\cite{Simone}.  A similar analysis can be performed for the TMSV state case, obtaining comparable results.

The main source of losses appears in the traveling phase of the protocol. Here, the idler must be preserved coherently in order to profit from the initial quantum correlations. In fact, it is easy to see that under the lossy dynamics described in Eq.~\eqref{losqub}, we have that
\begin{align}\label{noisySNR}
  \frac{Q^{dec}_{\hat O_{opt}}}{Q_{\hat O_{opt}}}=e^{-2\frac{t}{T_2}},
\end{align}
where $Q^{dec}_{\hat O_{opt}}$ is the SNR for the protocol applied to the state $\rho^{dec}= e^{t\mathcal{L}_D^q/\hbar}[|\psi\rangle_{\rm SC}\langle \psi|]$, $T_2^{-1}=\gamma + \frac{\Gamma_{\uparrow}+\Gamma_{\downarrow}}{2}$, $t$ is the traveling time, and we have discarded any initial thermal contributions (see Appendix~\hyperref[suppl:IV]{D}). This means that the protocol must be performed in a time well below the dephasing time of the qubit. Nowadays, qubits with $100$~$\mu$s lifetime can be realized~\cite{Houck}, corresponding to 30 km freespace propagation of light. An alternative would be the storage in high-Q Nb resonators. Presently, internal quality factors can reach above 1 milion~\cite{Cleland}, which at $5$~GHz frequencies it corresponds to a decay time of $2/5$~ms. Other options include highly coherent two-level systems formed in the junctions of qubits and high-frequency piezo-mechanical modes. 

Digital methods based on error correction have been widely studied in the context of quantum computing. There are principally two approaches to tackle the decoherence problem with a digital approach.  We may encode the idler into a logical qubit since the beginning, and perform the protocol in the logical Hilbert space. This is always doable in principle, and one may make a statement that an efficient cQED error-correction code implementation will be soon reached in the context of quantum computing~\cite{Google}. However, this approach is generally costy, as it requires the simultaneous control of several qubits. An alternative consists in exploiting the possibility of transfer the qubit information to the infinite degrees of freedom of a bosonic resonator field via a Jaynes-Cumming interaction~\cite{Qmem}. This approach requires only one resonator to store the idler, making the syndrome detection and error correction tasks easier to realize, because dissipation would be the main source of noise at $T\simeq 20$~mK. These so-called cat codes are at the basis of one of the most promising quantum computing architectures, and they have been experimentally demonstrated. Theoretically, one may reach substantial fidelity improvements over the uncorrected protocol for given time, with millisecond lifetime instead of hundred of microseconds of a bare transmon qubit~\cite{Qmem}. In principle, this approach should be better than using the Fock states of a resonator as a qubit. However, further experimental research is needed in this context, as the lifetime of the cat-qubit implemented in a recent experiment has been only 1.1 larger than an uncorrected qubit encoded in the Fock basis of a resonator~\cite{Ofek16}.

\section{VII. Conclusion}

We have developed the theory for performing a two-way covert quantum communication protocol in the low-frequency regime, in the case where one party has a severe energy constraint. While the results of this article are quite general, we have focused mainly on the $1-10$~GHz spectrum, where cQED platforms have been highly developed in the late decades. We have proved the ultimate bounds for the optimal receivers, finding that a quantum correlated detector can be at most a factor of four better in terms of SNR. Our bounds can be directly applied to the performance of a quantum illumination protocol. They imply that if the source is strongly amplified, then a protocol based on coherent state is optimal, ruling out any quantum advantage in a recent experiment~\cite{Shabir}. We have used the quantum illumination paradigm as a tool to perform two-way communication in the scenario where the sender is constrained to passive operations. Indeed, we have proved the square-root law for covert communication in our two-way setup, showing that $O(\sqrt{n})$ bits can be covertly transmitted by using the channel $n$ times. On the practical side, covertness is still limited by the low transmission rates. However, there are many machine type communications applications such as metering of electric power, gas and water that generate very little traffic per day and can tolerate large delays~\cite{Raza}. That is, their required data rates can be as low as millibits per second. These systems often transmit privacy sensitive data. Even if post-quantum encryption is utilized, the adversary could still learn privacy jeopardizing information simply by observing the traffic pattern~\cite{Hafeez}. Finally, we have provided the ingredients for performing a cQED based experiment, using Schr\"odinger's cat states as resource. Our implementation concept relies on qubit measurements instead of photodetection, notably improving the experimental requirements for the entanglement-assisted protocol. Developing a microwave quantum communication theory is a challenging task, due to the amount of noise that the related systems exhibit at room temperature. Indeed, our results contribute towards an implementation of open-air microwave quantum communication. In particular, we have settled, for the first time, a rigorous ground for developing a quantum-enhanced version of backscatter communication. This paradigm is now gaining plenty of interest in the communication engineering community due its ability of communicating with low-energy devices, with applications in RF communications, Internet-of-Things and NFC based technology~\cite{Griffin}. Indeed, our theoretical treatment gives a new twist to the field of backscatter communication, and we believe it will inspired further research on this line.

\section*{Acknowledgements}

The authors acknowledge support from Academy of
Finland under project no. 319578, under the RADDESS programme project no. 328193, and under the “Finnish Center of
Excellence in Quantum Technology QTF” Project Nos. 312296, 336810. RD acknowledges support from the Marie Sk{\l}odowska Curie
fellowship number 891517 (MSC-IF Green-MIQUEC). GSP acknowledges funding received from the European Union's Horizon 2020 research and innovation programme under grant agreement no. 862644 (FET-Open project QUARTET). We are also grateful for the support of the Scientiﬁc Advisory Board for Defence (Finland) and Saab.

 The authors thank Sergey N. Filippov, Giuseppe Vitagliano, G\"oran Johansson, Stefano Pirandola, and Kirill G. Fedorov for their useful comments to the manuscript in its various stages.

\clearpage
\widetext
\begin{center}
\textbf{\large Appendix: Two-way covert microwave quantum communication}

\end{center}

\setcounter{equation}{0}
\setcounter{figure}{0}
\setcounter{table}{0}
\setcounter{page}{1}
\setcounter{lemma}{0}

\makeatletter

\renewcommand{\theequation}{A\arabic{equation}}
\renewcommand{\thelemma}{A\arabic{lemma}}
\renewcommand{\thefigure}{A\arabic{figure}}
\renewcommand{\bibnumfmt}[1]{[A#1]}
\renewcommand{\citenumfont}[1]{A#1}
\renewcommand{\thepage}{A\arabic{page}}

In this Appendix, we provide the proofs of the results stated in the main text. We make frequent use of the following objects:
\begin{itemize}
    \item {\it Fock basis}, indicated with latin alphabet kets (or bra): $\{|k\rangle\}_{k=0}^\infty$; 
    \item {\it Coherent states} with amplitude $\alpha\in\mathbb{C}$, indicated with greek alphabet kets (or bra): $|\alpha\rangle= e^{-\frac{|\alpha|^2}{2}}\sum_{k=0}^\infty \frac{\alpha^k}{\sqrt{k!}}|k\rangle$;
    \item {\it Thermal States} with $N_B$ average photon numbers: $\rho_B=\sum_k \tau_k |k\rangle\langle k|$, where $\tau_k=\frac{1}{1+N_B}\left(\frac{N_B}{1+N_B}\right)^k$;
    \item {\it General signal-idler state} of $r$ Schmidt-rank: $|\psi\rangle_{SI}=\sum_{k=0}^r \sqrt{p_k}|v_k\rangle_I |w_k\rangle_S$. The signal mode is indicated by $\hat a_S$ and we use indistinctively the notation $|v\rangle|w\rangle$ and $|v,w\rangle$. The Schmidt rank $r$ differentiates between the entangled ($r>1$) and the idler-free ($r=1$) cases;
    \item {\it Constants}: $c_B=\frac{N_B}{1+N_B}$ and $c_S=\frac{N_S}{1+N_S}$.
    \end{itemize}

\section{Appendix~A: Receiver error probability}\label{suppl:I}

In this section, we discuss the results based on the calculation of the Chernoff bounds and the quantum Fisher information (QFI) relative to Bob's receiver. In the following, as Bob and Alice are sharing $\varphi_k$ values, we can set it to zero.

\subsection{1. Equivalence between OOK and BPSK}

In the $\eta\ll1$ limit, we can map the problem of discriminating between different $\phi$ to the quantum illumination (QI) setup, where Alice decides to modulate the amplitude between $\eta=0$ and $\eta=\bar \eta$, leaving the phase unchanged ($\phi=0$). This is usually referred as On-Off-Keying (OOK). We first notice that the received modes $\{\hat a_R^{(k)}\}$ can be expressed as
\begin{align}
    \hat a_R^{(k)}=\eta~(e^{-i\phi} \hat a_S^{(k)})+\sqrt{1-\eta^2}~\hat h^{(k)}, 
\end{align}
where $\hat h^{(k)}\equiv \sqrt{\frac{\eta}{1+\eta}}~e^{-i
\phi}\hat h_\leftarrow^{(k)}+\sqrt{\frac{1}{1+\eta}}~\hat h_\rightarrow^{(k)}$ are thermal modes with $N_B$ average number of photons. This means that if $|\psi\rangle_{SI}\langle \psi|$ is the state of the SI system, then the final Bob's state is $\rho_{\eta,\phi}=\text{Tr}_E\,\left[\hat B_\eta\hat U_\phi|\psi\rangle_{SI}\langle \psi|\otimes \rho_B \hat U_\phi^\dag\hat B_\eta^\dag\right]$, where $\hat B_\eta=\exp\left[\arccos \eta( \hat a_S^\dag \hat h - \hat a_S\hat h^\dag)\right]$ is a beamsplitter operation and $\hat U_\phi=e^{-i\phi \hat n_S}$ is a phase-shift operation. We can now prove the equivalence between OOK and BPSK.

\begin{lemma}{\bf[Equivalence of OOK and BPSK]}\label{lemma1}
In the $\eta\ll1$ limit, the BPSK and OOK optimal strategies are the same for both the local and the collective cases. The BPSK performs as an OOK with $\bar \eta=2\eta$. 
\end{lemma}
\begin{proof}
We have that $\rho_{\eta,0}=\rho_{B}-\eta d\rho +o(\eta)$, $\rho_{\eta,\pi}=\rho_B+\eta d\rho +o(\eta)$ for $\eta\ll1$, where $d\rho=\text{Tr}_
S\,[\hat a_S^\dag \hat h-\hat a_S\hat h^\dag, |\psi\rangle_{SI}\langle \psi|\otimes \rho_B]$. Therefore, $\rho_{\eta,0}^{\otimes n}=\rho_B^{\otimes n}-\eta d\sigma +o(\eta)$ and $\rho_{\eta,\pi}^{\otimes n}=\rho_B^{\otimes n}+\eta d\sigma +o(\eta)$, with $d\sigma=\sum_{i=1}^n \rho_B^{\otimes j-1}\otimes d\rho \otimes \rho_B^{\otimes n-j}$. 
This means that $(\rho_{\eta,\pi}^{\otimes n}-\rho_{\eta,0}^{\otimes n})=(\rho_{0,0}^{\otimes n}-\rho_{2 \eta,0}^{\otimes n})+o(\eta)$. Therefore, BPSK performs as an OOK with $\bar \eta=2\eta$ in the $\eta\ll1$ limit, and their measurement setups - being local or collective - are the same. 
\end{proof}

\subsection{2. Chernoff bound and quantum Fisher information: general formulas}

We can now analyze the OOK case to state the general formulas for the quantum Chernoff bound and quantum Fisher information for the BPSK case. In the following, we denote $\rho_\eta\equiv \rho_{\eta,0}$.

\begin{lemma}{\bf [Chernoff bound for QI (OOK)]}\label{Chernoff}
Given $\rho_{\eta}= \text{Tr}_E\, (\hat B_{\eta}|\psi\rangle_{SI}\langle \psi| \otimes \rho_B \hat B_{\eta}^\dag)$, with $\hat B_\eta = \exp{\left[\arccos\sqrt{\eta}(\hat a_S^\dag \hat h -\hat a_S \hat h^\dag)\right]}$. Then, the optimal error probability in the task of distinguishing between $\rho_0$ and $\rho_{\bar \eta}$ is $p_{\rm err}\leq \frac{1}{2} e^{-MC(\rho_0,\rho_{\bar \eta})}$ for $M\gg1$, where
\begin{equation}
C(\rho_0,\rho_{\bar \eta}) = \frac{\bar \eta^2}{1+N_B} \sum_{k,k'}\frac{p_{k}p_{k'}|\langle w_{k'}|\hat a_S|w_{k}\rangle |^2}{\left[\sqrt{p_{k'}}+\sqrt{p_{k}}\sqrt{c_B}\right]^2} + o(\bar \eta ^2).\label{chern}
\end{equation}
\end{lemma}
\begin{proof}
This is a simple application of one of the results in Calsamiglia et al.~\cite{S_Calsamiglia08}. We have that $C(\rho_0,\rho_{\bar \eta})=-\min_{s\in[0,1]}\log \text{Tr}(\rho_0^s \rho_{\bar \eta}^{1-s})$. Considering the Taylor expansion around $\bar \eta=0$, $\rho_{\bar \eta}=\rho_0 + \bar \eta d\rho+o(\bar \eta)$, then
\begin{align}
    C(\rho_0,\rho_{\bar \eta}) &= \frac{\bar \eta^2}{2}\sum_{kk'nn'}\frac{|\langle v_k, n |d\rho |v_{k'}n'\rangle|^2}{[\sqrt{p_{k}\tau_n}+\sqrt{p_{k'}\tau_{n'}}]^2} + o(\bar \eta^2)\\
    \quad &=\bar \eta^2\beta^{col} +o(\bar \eta^2), \label{Cher1}
\end{align}
see Equation~(47) of Ref.~\cite{S_Calsamiglia08}. The task reduces in computing Eq.~\eqref{Cher1} with $d\rho=\text{Tr}_S [\hat a_S^\dag \hat h-\hat a_S \hat h^\dag,|\psi\rangle_{SI}\langle\psi|\otimes \rho_{B}]$. First, we notice that 
\begin{equation}
\langle v_k, n |d\rho |v_{k'}n'\rangle = (\tau_{n'}-\tau_n)[\langle w_{k'}|\hat a_S^\dag | w_k\rangle\sqrt{n+1}\delta_{n',n+1}-\langle w_{k'}|\hat a_S| w_{k}\rangle\sqrt{n'+1}\delta_{n,n'+1}].
\end{equation}
We have that
\begin{align}
\beta^{col}&=\frac{1}{2}\sum_{k k' n n'}\frac{|\langle v_k,n|\sum_{jj'}\sqrt{p_j p_{j'}} |v_j\rangle \langle v_{j'}|\otimes [\langle w_{j'}|\hat a_S^\dag|w_{j}\rangle \hat h - \langle w_{j'}|\hat a_S | w_{j}\rangle \hat h^\dag, \rho_{B}] |v_{k'},n'\rangle|^2}{[\sqrt{p_k\tau_n}+\sqrt{p_{k'}\tau_{n'}}]^2} \\
\quad&=\frac{1}{2}\sum_{k k' n n'} \frac{p_{k}p_{k'}|\langle w_{k'}|\hat a_S^\dag|w_{k}\rangle(\tau_{n'}-\tau_n)\sqrt{n+1}\delta_{n',n+1}- \langle w_{k'}|\hat a_S|w_{k}\rangle(\tau_{n'}-\tau_n)\sqrt{n'+1}\delta_{n,n'+1} |^2}{[\sqrt{p_k\tau_n}+\sqrt{p_{k'}\tau_{n'}}]^2}\\
\quad&=\frac{1}{2}\sum_{k k' n n'}\frac{p_{k}p_{k'}(\tau_{n'}-\tau_n)^2}{[\sqrt{p_k\tau_n}+\sqrt{p_{k'}\tau_{n'}}]^2}\left[|\langle w_{k'}|\hat a_S^\dag|w_{k}\rangle|^2(n+1)\delta_{n',n+1} + |\langle w_{k'}|\hat a_S |w_{k}\rangle|^2 (n'+1)\delta_{n,n'+1}\right] \\
\quad&=\frac{1}{2}\sum_{k k' n} (n+1)p_{k}p_{k'}[\tau_{n+1}-\tau_n]^2 \left[\frac{|\langle w_{k'}|\hat a_S^\dag|w_{k}\rangle|^2}{[\sqrt{p_k\tau_n}+\sqrt{p_{k'}\tau_{n+1}}]^2}+\frac{|\langle w_{k'}|\hat a_S|w_{k}\rangle|^2}{[\sqrt{p_k\tau_{n+1}}+\sqrt{p_{k'}\tau_{n}}]^2} \right],
\end{align}
where in the last line we have summed on the $n'$ index. We now use that the last sum is symmetric under the exchange of $k$ and $k'$ and that  $\tau_n/\tau_{n-1}=\frac{N_B}{1+N_B}$:
\begin{align}
\beta^{col}&=\sum_{k k' n} (n+1)p_{k}p_{k'}[\tau_{n+1}-\tau_n]^2 \frac{|\langle w_{k'}|\hat a_S|w_{k}\rangle|^2}{[\sqrt{p_k\tau_{n+1}}+\sqrt{p_{k'}\tau_{n}}]^2})  \\
\quad& =  \sum_{k k' n} (n+1)\tau_np_{k}p_{k'}\left[1-\frac{\tau_{n+1}}{\tau_n}\right]^2 \frac{|\langle w_{k'}|\hat a_S|w_{k}\rangle|^2}{\left[\sqrt{p_{k'}}+\sqrt{p_{k}}\sqrt{\frac{\tau_{n+1}}{\tau_n}}\right]^2} \\
\quad&= \frac{1}{1+N_B}\sum_{kk'}\frac{p_{k}p_{k'}|\langle w_{k'}|\hat a_S|w_{k}\rangle|^2}{\left[\sqrt{p_{k'}}+\sqrt{p_{k}}\sqrt{\frac{N_B}{1+N_B}}\right]^2}. \label{chercalc}
\end{align}
\end{proof}
    
\begin{lemma}{\bf [Quantum Fisher information for QI (OOK)~\cite{S_Sanz17}]}\label{lemmaQFI}
Given $\rho_{\eta}$ as in Lemma~\ref{Chernoff}. Then, the quantum Fisher information for estimating the parameter $\eta$ in the $\eta\ll1$ neighborhood is
\begin{equation}\label{QF}
F = \frac{4}{1+N_B} \sum_{k,k'}\frac{p_{k}p_{k'}}{p_{k'}+p_{k}c_B}|\langle w_{k'}|\hat a_S|w_{k}\rangle |^2.
\end{equation}

\end{lemma}
\begin{proof}
The QFI is given by~\cite{S_Paris09}
\begin{equation}
    F=2\sum_{kk' nn'}\frac{|\langle v_k, n |d\rho |v_{k'}n'\rangle|^2}{p_{k}\tau_n+{p_{k'}\tau_{n'}}},
\end{equation}
where $d\rho=\text{Tr}_S[\hat a_S^\dag \hat h-\hat a_S \hat h^\dag,|\psi\rangle_{SI}\langle\psi|\otimes \rho_{B}]$. The calculation is similar as in Lemma~\ref{Chernoff}.
\end{proof}

\subsection{3. Ultimate error probability bounds for the receiver}
We can now prove some of the Theorems of section~\hyperref[sec:III]{III} the main text.

\begin{proof}[Proof of Theorem~\ref{QCB}]
For the idler-free case, we apply Lemma~\ref{Chernoff} with $\bar \eta=2\eta$ (see Lemma~\ref{lemma1}) to the simple case of Schmidt-rank one, finding that $\beta^{loc}_{cl}=\frac{4|\langle w|\hat a_S| w\rangle|^2}{1+N_B}\frac{1}{(1+\sqrt{c_B})^2}$. Then, by applying the H\"older's inequality, we find that  $|\langle w|a_S| w\rangle|^2\leq \| a_S|w\rangle\|_2^2= N_S$, which is saturated by $|w\rangle=|\alpha\rangle$.
For the general case, by applying the inequality $\frac{p_{k'}}{\left[\sqrt{p_{k'}}+\sqrt{p_{k}}\sqrt{c_B}\right]^2}\leq 1$ to Eq.~\eqref{chern} with $\bar \eta=2\eta$ (see Lemma~\ref{lemma1}), we obtain $\beta^{col}\leq \frac{4}{1+N_B}\sum_{k,k'}p_{k}\langle w_{k}|\hat a_S^\dag |w_{k'}\rangle\langle w_{k'}|\hat a_S|w_{k}\rangle$. By using the completeness relation $\sum_{k'}|w_{k'}\rangle\langle w_{k'}|=\mathbb{I}$ -  which can be assumed by adding zero probability terms to the sum - and by noticing that $N_S=\sum_{k}p_{k}\langle w_{k}|\hat a_S^\dag \hat a_S |w_k\rangle$, we conclude that $\beta^{col}\leq \frac{4N_S}{1+N_S}$. 
By applying the inequality of arithmetic and geometric means $\frac{p_{k} p_{k'}}{[\sqrt{p_{k'}}+\sqrt{p_{k}}c_B]^2}\leq \frac{\sqrt{p_{k}p_{k'}}}{4\sqrt{c_B}}\leq \frac{p_{k}+p_{k'}}{8\sqrt{c_B}}$, and by using the completeness relation, we find the second inequality $\beta^{col}\leq \frac{2N_S+1}{1+N_B}\frac{1}{4\sqrt{c_B}}$. Moreover, no mixed state can do better, as in this case the bound can be applied to its purification.
\end{proof}

\begin{proof}[Proof of Theorem~\ref{QFIbound}]
For the idler-free case, by applying Lemma~\ref{lemmaQFI} with $\bar \eta=2\eta$ (seel Lemma~\ref{lemma1}) to the Schmidt-rank one case, we find that $\beta_{cl}^{loc}=\frac{2|\langle w|\hat a_S| w\rangle|^2}{1+N_B}\frac{1}{1+c_B}$, which is maximal for $|w\rangle=|\alpha\rangle$ (see the proof of Theorem~\ref{QCB}). Homodyne is optimal as one can directly see by checking that the signal-to-noise ratio $\langle\hat x\rangle^2_{\rho_{\bar\eta}}/\langle\hat x^2\rangle_{\rho_0}$ saturates the QFI, and by using Lemma~\ref{lemma1}.
The general bound is found similarly as in the proof of Theorem~\ref{QCB}, with the inequalities $\frac{p_{k'}}{p_{k'}+p_{\alpha}c_B}\leq 1$ and $\frac{p_{k}p_{k'}}{p_{k'}+p_{k}c_B}\leq \frac{p_{k}+p_{k'}}{4\sqrt{c_B}}$ applied to Eq.~\eqref{QF}. Also in this case no mixed state can do better, by applying the bound to the purified state.

\end{proof}

\subsection{4. Examples}

{\it QCB and QFI of TMSV states}: This is done by setting $p_{k}=\frac{1}{1+N_S}c_S^k$ and $|w_k\rangle=|k\rangle$ (Fock state with $k$ photons) into the Eq.~\eqref{chern}, where we set $\bar \eta=2\eta$, and Eq.~\eqref{QF}. It results in a sum of a geometric series and its first derivative, that can be cast as written in the main text. Similarly, the optimal observable for the threshold discrimination strategy is found by computing $\sum_{k k' n n'}\frac{\langle k',n'|d\rho|k,n\rangle}{p_{k'}\tau_{n'}+p_{k}\tau_{n}}|k', n'\rangle\langle k,n|$~{\it (32)}.

{\it QCB and QFI of SC states}: This is done by applying Lemma~\ref{Chernoff} and \ref{lemmaQFI} to the Schmidt decomposition of the SC state given in Eq.~\eqref{SCShmidt} of the main text, i.e. $|\psi\rangle_{\rm SC}= \sqrt{\lambda_+}|g\rangle|\alpha_+\rangle + \sqrt{\lambda_-}|e\rangle|\alpha_-\rangle$. We then use Lemma~\ref{lemma1} to bring the result to the BPSK case. The result is
\begin{align}
\beta_{\rm SC}^{col} = &\frac{N_S}{1+N_B}f^{col}_{\rm SC}(N_S,N_B)\\
\beta_{\rm SC}^{loc} = &\frac{N_S}{1+N_B}f^{loc}_{\rm SC}(N_S,N_B)
\end{align}
with 
\begin{align}
f^{col}_{\rm SC}(N_S,N_B)&=\frac{4\lambda_+^2}{\left(\sqrt{\lambda_+}+\sqrt{\lambda_-}\sqrt{c_B}\right)^2}+\frac{4\lambda_-^2}{\left(\sqrt{\lambda_-}+\sqrt{\lambda_+}\sqrt{c_B}\right)^2},\\
f^{col}_{\rm SC}(N_S,N_B)&= \frac{2\lambda_+^2}{\lambda_++\lambda_-c_B}+\frac{2\lambda_-^2}{\lambda_-+\lambda_+c_B}.
\end{align}
In the $N_B\gg1$ limit, these quantities approximate to 
\begin{align}
 f^{col}_{\rm SC} &\stackrel{N_B\gg1}{=}\frac{2+2e^{-4N_S}}{1+\sqrt{1-e^{-4N_S}}}=4-8\sqrt{N_S} +O(N_S),  \\
f^{loc}_{\rm SC} &\stackrel{N_B\gg1}{=}1+e^{-4N_S}=2-4N_S+O(N_S^2),
\end{align}
where we have used that $\lambda_{\pm}=\frac{1}{2}[1\pm e^{-2N_S}]$. The optimal local observable can be found by computing 
\begin{equation}
\hat O_{opt}=\sum_{k k'\in \{g, e\};\;n n'\in [0,\infty]}\frac{\langle k',n'|d\rho|k,n\rangle}{p_{k'}\tau_{n'}+p_{k}\tau_{n}}|k', n'\rangle\langle k,n|,
\end{equation}
where $p_e=\lambda_-$ and $p_g=\lambda_+$. Alternatively, one can directly compute the SNR of $\hat O_{opt}$ and see that it saturates the QFI in the $N_B\gg1$ limit.

\section{Appendix~B: Covert Quantum Communication}\label{suppl:II}

Here, we provide the technical details to prove the main results on covert quantum communication. We denote Eve's quantum state when Alice applies a phase modulation $\tilde \varphi$ as $\rho^{(N_S)}_{\tilde \varphi}$. $N_S>0$ corresponds to the $on$-setting, while $N_S=0$ is the ${\it off}$-setting. We drop any $k$ superscript and subscript, as we are in the i.i.d assumptions. In addition, we introduce the beamsplitter unitary operator $\hat B_{12}=\exp{\left[\theta(\hat a_1 \hat a_2^\dag-\hat a_1^\dag \hat a_2)\right]}$, where $\theta=\arccos\sqrt{\eta}$. Let us introduce 
\begin{align}
  \mathcal{E}_{\tilde \varphi}[\sigma]&=\text{Tr}_S\, \left[\hat B_{\rightarrow, S}e^{-i\tilde \varphi \hat n_S}\hat B_{\leftarrow,S} (\rho_B\otimes\rho_B\otimes\sigma)\hat B_{\leftarrow,S}^\dag e^{i\tilde \varphi \hat n_S} \hat B_{\rightarrow, S}^\dag\right]\\
  \quad&= e^{-i\tilde \varphi \hat n_\rightarrow}\text{Tr}_S\, \left[\hat B_{\rightarrow, S}\hat B_{\leftarrow,S} (\rho_B\otimes\rho_B\otimes\sigma)\hat B_{\leftarrow,S}^\dag \hat B_{\rightarrow, S}^\dag\right] e^{i\tilde \varphi \hat n_\rightarrow}, \label{beamspl}
\end{align}
 where $\sigma=\sum_{k,k'}c_{kk'}|k\rangle_S\langle k'|$, $\text{Tr}_S$ denotes the partial trace on the signal mode, and the equality is due to the phase-invariance of the thermal state. The latter is evident at seeing the input-output relations in Eqs.~\eqref{on1}-\eqref{on2} of the main text. Let us denote by $\rho^{(N_S)}=\frac{1}{|\mathcal{A}|}\sum_{\tilde \varphi\in\mathcal{A}} \rho_{\tilde \varphi}^{(N_S)}$, where $\rho_{\tilde \varphi}^{(N_S)}=\mathcal{E}_{\tilde\varphi}[\rho_S]$. Here, we denote the phase shift at Alice as $\tilde\varphi$, which is $\varphi+\phi$ or $\varphi$ depending if we are in the $on$ or ${\it off}$ setting respectively. In both cases, $\tilde\varphi$ is distributed uniformly at random in $\mathcal{A}$.  

\begin{proof}[Proof of Lemma~\ref{theocovert}] \mbox{}\\*
{\bf [Achievability]} We have that $P^{({\rm Eve})}= \frac{1}{2}\left[1-\frac{1}{2}\| {\rho^{(N_S)}}^{\otimes n}-{\rho^{(0)}}^{\otimes n}\|_1\right]$~\cite{S_Calsamiglia08}. As done in Ref.~\cite{S_Bash15}, we can simplify the calculation by using the Pinsker's inequality, i.e. $\| \rho_a - \rho_b\|_1\leq \sqrt{2 D(\rho_a,\rho_b)}$ for any states $\rho_a$ and $\rho_b$, where $D(\rho_a,\rho_b)=-\text{Tr }\, \rho_a \ln \rho_b+\text{Tr}\, \rho_a\ln \rho_a$ is the quantum relative entropy between $\rho_a$ and $\rho_b$. This provides the bound 
\begin{align}
    P^{({\rm Eve})}\geq \frac{1}{2}-\sqrt{\frac{1}{8}D({\rho^{(0)}}^{\otimes n},{\rho^{(N_S)}}^{\otimes n})},
\end{align}
meaning that 
\begin{equation}\label{rel}
    D({\rho^{(0)}}^{\otimes n},{\rho^{(N_S)}}^{\otimes n})\leq 8\delta^2 
\end{equation}
ensures that $P^{({\rm Eve})}\geq \frac{1}{2}-\delta$ over $n$ modes. We use that the quantum relative entropy is additive for tensor product, i.e.   $D({\rho^{(0)}}^{\otimes n},{\rho^{(N_S)}}^{\otimes n})=nD({\rho^{(0)}},  {\rho^{(N_S)}})$ to reduce the calculation to the single channel-usage case. 

\begin{itemize}
    \item {\bf TMSV case:} Let us define the vector $\vec{r}=(\hat x_\leftarrow,\hat p_\leftarrow, \hat x_\rightarrow, \hat p_\rightarrow)^{T}$, where $\hat x_l=\frac{\hat a_l+\hat a_l^\dag}{\sqrt{2}}$ and $\hat p_l=\frac{\hat a_l-\hat a_l^\dag}{\sqrt{2}i}$ ($l\in\{\leftarrow,\rightarrow\}$). For a zero-mean Gaussian state $\rho$, the covariance matrix is defined as $\Sigma_{ij}=\text{Tr}\,(\{\hat r_j,\hat r_k\}\rho)$. The covariance matrix of the Gaussian state $\rho_{\tilde \varphi}^{(N_S)}$ is
\begin{equation}
\Sigma_{\tilde \varphi}^{(N_S)}=2\times
\begin{bmatrix}
A & 0 & -B \cos \tilde \varphi &  B \sin \tilde \varphi\\
0 & A& -B \sin \tilde \varphi & -B \cos \tilde \varphi \\
-B \cos \tilde \varphi & -B\sin \tilde \varphi & C & 0 \\
B \sin \tilde \varphi & -B \cos \tilde \varphi & 0 & C
\end{bmatrix}
\end{equation}
where $A= \frac{1}{2}+\eta N_B +(1-\eta)N_S$, $B=(1-\eta)\sqrt{\eta} (N_B-N_S)$, and $C=\frac{1}{2}+[(1-\eta)^2+\eta]N_B+(1-\eta)\eta N_S$. We have that \begin{align}
    D_{\rm Gauss}&=D(\rho^{(0)},\rho^{(N_S)})\\
    \quad&\leq \frac{1}{|\mathcal{A}|}\sum_{\tilde \varphi\in\mathcal{A}} D(\rho_{\tilde \varphi}^{(0)},\rho_{\tilde \varphi}^{(N_S)}) \\
    \quad &= D(\rho_{\tilde \varphi=0}^{(0)},\rho_{\tilde \varphi=0}^{(N_S)}),
\end{align}
where we have used the joint convexity property of the relative entropy, and that the $D(\hat U\rho \hat U^\dag,\hat U\sigma \hat U^\dag)=D(\rho,\sigma)$ for any unitary $\hat U$ and states $\rho$ and $\sigma$, together with Eq.~\eqref{beamspl}. For zero-mean Gaussian states, we have that 
\begin{align}\label{relGaussian}
    D(\rho_{\tilde \varphi}^{(0)},\rho_{\tilde \varphi}^{(N_S)}) &= \frac{1}{2}\left[\log\frac{\text{det}[\Sigma_{\tilde \varphi}^{(N_S)}+i\Omega]}{\text{det}[\Sigma_{\tilde \varphi}^{(0)}+i\Omega]}+\frac{1}{2}\text{Tr}\,[\Sigma_{\tilde \varphi}^{(0)}(H_{\tilde \varphi}^{(N_S)}-H_{\tilde \varphi}^{(0)})]\right],
\end{align}
where $\Omega=-i\mathbb{I}_2\otimes\sigma_y$ is the symplectic form and $H_{\tilde \varphi}^{(N_S)}=2~\text{arccoth}(i\Omega \Sigma_{\tilde \varphi}^{(N_S)})i\Omega$ is the Hamiltonian matrix, given that $\rho_{\tilde\varphi}^{(N_S)}$ is a faithful Gaussian state~\cite{S_PirandolaEntropy, S_Wilde17}. Eq.~\eqref{relGaussian} has been computed using Mathematica, finding 
\begin{align}
    D(\rho_{\tilde \varphi}^{(0)},\rho_{\tilde \varphi}^{(N_S)}) =
    &-(1 + 2 N_B \eta^2+2N_S(1-\eta^2)) \left[\text{arccoth}(1 + 2 N_B \eta^2) - 
    \text{arccoth}(1 + 2 N_S + 2 (N_B - N_S) \eta^2)\right] \nonumber \\
    \quad&+ 
 \frac{1}{2} \log\frac{N_B \eta^2 (1 + N_B \eta^2)}{N_S (1 + N_S) + (N_B - N_S) (1 + 2 N_S) \eta^2 + (N_B - N_S)^2 \eta^4},
\end{align}
which, as already mentioned, does not depend on $\tilde \varphi$. The expansion to the third order in $N_S$ gives 
\begin{align}
 D(\rho_{\tilde \varphi=0}^{(N_S)},\rho_{\tilde \varphi=0}^{(0)}) &= \frac{(1-\eta^2)^2}{2N_B\eta^2(1+N_B\eta^2)}N_S^2+ aN_S^3+o(N_S^3)\\
     \quad&\leq \frac{(1-\eta^2)^2}{2N_B\eta^2(1+N_B\eta^2)}N_S^2 \label{eqq}
\end{align}
where $a<0$ allows us to use the Taylor's remainder theorem.

\item {\bf General case:} We extend the result to signal states of the form $\rho_S=\sum_{j=0}^\infty N_S^j\sigma_j$, where 
\begin{align}
    \sigma_0&= |0\rangle\langle0|\\
    \sigma_1&= |1\rangle\langle1|-|0\rangle\langle0| \\
    \sigma_2&= c(|2\rangle\langle2|-2|1\rangle\langle1|+|0\rangle\langle0|),
\end{align}
with $0\leq c\leq 1$. This set of states include the Gaussian state case ($c=1$). In addition, for $N_S\ll1$ these states well approximate Gaussian states. 
  Let us consider the Taylor expansion for the logarithm of a matrix 
 \begin{align}
    \log(A+tB)&= \log(A)+t\int_0^\infty \frac{1}{A+z}B\frac{1}{A+z}dz-t^2\int_0^\infty \frac{1}{A+z}B\frac{1}{A+z}B\frac{1}{A+z}dz \nonumber \\ 
     \quad & \quad +t^3\int_0^\infty \frac{1}{A+z}B\frac{1}{A+z}B\frac{1}{A+z}B\frac{1}{A+z}dz + +o(t^3).
 \end{align} 
 If we set $A=\rho^{(0)}$, $B=\frac{\rho^{(N_S)}-\rho^{(0)}}{N_S}$ and $t=N_S$, we obtain
\begin{align}    D_c(\rho^{(0)},\rho^{(N_S)})&=-\text{Tr}\,\rho^{(0)}\ln \rho^{(N_S)}+\text{Tr}\,\rho^{(0)}\log \rho^{(0)} \\
   \quad& = -\text{Tr}\,(\rho^{(N_S)}-\rho^{(0)})\int_0^\infty \frac{1}{\rho^{(0)}+z}\rho^{(0)}\frac{1}{\rho^{(0)}+z}dz \nonumber\\
   \quad  &\quad +\text{Tr}\,\rho^{(0)}\int_0^\infty \frac{1}{\rho^{(0)}+z}(\rho^{(N_S)}-\rho^{(0)})\frac{1}{\rho^{(0)}+z}(\rho^{(N_S)}-\rho^{(0)})\frac{1}{\rho^{(0)}+z}dz \nonumber \\ 
     \quad & \quad -\text{Tr}\,\rho^{(0)}\int_0^\infty \frac{1}{\rho^{(0)}+z}(\rho^{(N_S)}-\rho^{(0)})\frac{1}{\rho^{(0)}+z}(\rho^{(N_S)}-\rho^{(0)})\frac{1}{\rho^{(0)}+z}(\rho^{(N_S)}-\rho^{(0)})\frac{1}{\rho^{(0)}+z}dz \nonumber\\
     \quad&\quad+o(N_S^3) \label{firstline} \\
     \quad &= N_S^2\text{Tr}\,\rho^{(0)}\int_0^\infty \frac{1}{\rho^{(0)}+z}\rho_1\frac{1}{\rho^{(0)}+z}\rho_1\frac{1}{\rho^{(0)}+z}dz \quad\quad (=b_1N_S^2) \nonumber \\
     \quad& \quad -N_S^3 \text{Tr}\,\rho^{(0)}\int_0^\infty \frac{1}{\rho^{(0)}+z}\rho_1\frac{1}{\rho^{(0)}+z}\rho_1\frac{1}{\rho^{(0)}+z}\rho_1\frac{1}{\rho^{(0)}+z}dz  \quad\quad (=b_2 N_S^3)\nonumber\\
     \quad &\quad +cN_S^3\text{Tr}\,\rho^{(0)}\int_0^\infty \left[\frac{1}{\rho^{(0)}+z}\rho_1\frac{1}{\rho^{(0)}+z}\rho_2\frac{1}{\rho^{(0)}+z} + \frac{1}{\rho^{(0)}+z}\rho_2\frac{1}{\rho^{(0)}+z}\rho_1\frac{1}{\rho^{(0)}+z} \right]dz\quad (=cb_3 N_S^3) \nonumber \\
     \quad&\quad+o(N_S^3),
  \end{align}  
  where $\rho_k=\frac{1}{|\mathcal{A}|}\sum_{\tilde \varphi\in\mathcal{A}}\mathcal{E}_{\tilde \varphi}[\sigma_k]$. Here, we have used that $\int_0^\infty \frac{s}{(s+z)^2}dz=1$ and that $\rho^{(N_S)}-\rho^{(0)}$ is traceless in order to conclude that the first line of Eq.~\eqref{firstline} is zero.  Next, we prove that $b_2\leq 0$. We have that $b_2=-\text{Tr}\,\int_0^\infty A_zB_zA_zB_zA_zdz$,  with $A_z=\frac{\sqrt{\rho^{(0)}}}{\rho^{(0)}+z}\rho_1\frac{\sqrt{\rho^{(0)}}}{\rho^{(0)}+z}$ and $B_z=\frac{\rho^{(0)}+z}{\rho^{(0)}}$, as $\rho^{(0)}$ is full-rank. Therefore, we have the bound
  \begin{align}
      b_2&= -\text{Tr}\,\int_0^\infty (A_zB_z)^2A_zdz \\
        \quad&\leq-\text{Tr}\, \int_0^\infty A_zdz =0,
  \end{align}
  where we have used that $(A_zB_z)^2\geq0$ and that $\rho_1$ is traceless. 
  
  We can now bound $D_c$ regardless of the sign of $b_3$, by using Eq.~\eqref{eqq}, as $c=1$ includes the Gaussian case. In fact, assume that $b_3\geq0$, then $D_c=D_{{\rm Gauss}}+(c-1)b_3 N_S^3+o(N_S^3)$. We can use Eq.~\eqref{eqq} and the Taylor's remainder theorem to conclude that $D_{0\leq c\leq1}\leq \frac{(1-\eta^2)^2N_S^2}{2N_B\eta^2(1+N_B\eta^2)}$. Assume that $b_3<0$, then we have that  $D_c\leq b_1 N_S^2$ by the Taylor's remainder theorem. In addition, we have that the bound $D_{\rm Gauss}=b_1N_S^2+O(N_S^3)\leq \frac{(1-\eta^2)^2N_S^2}{2N_B\eta^2(1+N_B\eta^2)}$ holds for any $N_S>0$, which implies that $b_1\leq \frac{(1-\eta^2)^2}{2N_B\eta^2(1+N_B\eta^2)}$.  

\end{itemize}

Therefore, if we choose $N_S\leq \frac{4\sqrt{N_B\eta^2(1+N_B\eta^2)}\delta}{(1-\eta^2)^2\sqrt{n}}$, then we have that $P^{{\rm (Eve)}}\geq\frac{1}{2}-\delta$ over $n$ channel usages.

{\bf [Converse]} We now prove that $D({\rho^{(N_S)}}^{\otimes n},{\rho^{(0)}}^{\otimes n})\leq 8\delta^2$ implies that $N_S\leq \frac{4\sqrt{\eta^2N_B(1+\eta^2N_B)}}{(1-\eta^2)^2}\frac{\delta}{\sqrt{n}}$. Let us consider Eve's modes $\hat w_\leftarrow^{(k)}$ and $\hat w_\rightarrow^{(k)}$, $k=1,\dots,n$. Let us these modes be the output of a black-box which has the knowledge of the individual realizations of $\tilde \varphi_k$. The black-box acts as a beamsplitter, generating the modes
\begin{align}
\hat w_1^{(k)}&=\sqrt{\frac{\eta}{1+\eta}}~\hat w_{\rightarrow}^{(k)}+\sqrt{\frac{1}{1+\eta}}~e^{-i\tilde\varphi_k}~\hat w_{\leftarrow}^{(k)} \\
\hat w_2^{(k)}&=-\sqrt{\frac{1}{1+\eta}}~\hat w_{\rightarrow}^{(k)}+ \sqrt{\frac{\eta}{1+\eta}}~e^{-i\tilde\varphi_k}~\hat w_{\leftarrow}^{(k)}.
\end{align}
We then trace-out the modes $\hat w_2^{(k)}$. Notice that 
\begin{equation}
\hat w_1^{(k)} = -\sqrt{1-\eta^2}~\hat a_S^{(k)} e^{-i\tilde \varphi_k} +\eta~\hat h^{(k)},
\end{equation}
where $\hat h^{(k)}=\sqrt{\frac{\eta}{1+\eta}}~\hat h_{\rightarrow}^{(k)}+\sqrt{\frac{1}{1+\eta}}~e^{-i\tilde\varphi_k}~\hat h_{\leftarrow}^{(k)}$ is in a thermal state with $N_B$ average number of photons, regardless of the value of $\tilde \varphi_k$. Since Eve does not have the knowledge of the phase $\tilde \varphi_k$ and $\hat a_S^{(k)}$ are i.i.d., the state of the modes $\hat w_1^{(k)}$ does not depend on $k$. Let us denote its density matrix as  $\sigma^{(N_S)}$. Here, $N_S=0$ in the ${\it off}$-setting and $N_S>0$ in the $on$-setting. We have that
\begin{equation}
D({\rho^{(0)}}^{\otimes n},{\rho^{(N_S)}}^{\otimes n})\geq D({\sigma^{(0)}}^{\otimes n},{\sigma^{(N_S)}}^{\otimes n}),
\end{equation}
which comes from the monotonicity of the quantum relative entropy under CPTP maps. This reduces the calculation to the one-way case. Following the proof of Theorem~1 in Ref.~\cite{Bullock20}, we find that
\begin{equation}
D({\sigma^{(0)}}^{\otimes n},{\sigma^{(N_S)}}^{\otimes n})\geq \frac{n(1-\eta^2)^2N_S^2}{2\eta^2N_B(1+\eta^2 N_B)}+o(N_S^2).
\end{equation}
Solving for $N_S$ ends the proof.
\end{proof}

Lemma~\ref{theocovert} implies the square-root law, provided that Alice and Bob share a codebook. Notice that the error probability of reading one wrong bit is
\begin{align}
    P_{\rm err}= 1-(1-p_{\rm err})^m\leq mp_{\rm err}
\end{align}
where $p_{\rm err}$ is the single-bit receiver error probability. This automatically means that a number $m=O(\sqrt{n}/\log n)$ bits are reliably transmissible. In fact, we have that $p_{\rm err}\leq \frac{1}{2}\exp{(-M\beta\eta^2 \frac{N_S}{1+N_B})}$ for $M$ large enough. Here, $\beta=4$ for the TMSV state and SC state transmitters with the optimal collective receiver, $\beta=2$ for the TMSV state and SC state transmitter with the optimal local receiver, and $\beta=1$ for the coherent state transmitter with a homodyne detector receiver. By setting $N_S=\frac{4\sqrt{N_B\eta^2(1+N_B\eta^2)}\delta}{(1-\eta^2)^2\sqrt{n}}$, with $n=mM$, we get
\begin{align}
    P_{\rm err} \leq \frac{m}{2}\exp \left(-4c_B\beta \delta \eta^4  \frac{\sqrt{n}}{m}\right),
\end{align}
where $c_B=\frac{N_B}{1+N_B}$. By setting $m=\frac{A \sqrt{n}}{\log\frac{A}{\epsilon}\log\sqrt{n}}$ with $A=4\delta \eta^4c_B\beta$, we have that $P_{\rm err}\leq \frac{\epsilon}{\log\frac{A}{\epsilon}\log\sqrt{n}}\leq\epsilon$, for small enough $\epsilon$. 

We can use the results in Refs.~\cite{S_Bash15, S_Bash13} for AWGN channels in order to get a better scaling for the decoding error probability. This is Theorem~\ref{SRL} of the main text.

\begin{proof}[Proof of Theorem~\ref{SRL}]
Let us define $\sigma_\beta^2=\frac{1+N_B}{2\beta \eta^2 M}$. In the optimal local protocol case, the induced AWGN channel for $M\gg1$ has a variance $\sigma_{\beta=2}^2$. For the coherent state case, the induced AWGN channel for any $M$ has a variance $\sigma_{\beta=1}^2$. In the optimal collective protocol case, we can induce a AWGN channel by dividing $M$ into $K\gg1$ slots of $M/K$ samples each, and apply the optimal collective protocol on each of the $K$ slots. This provides the same asymptotic performance for the receiver as long as $M/K\gg1$. The resulting AWGN channel has $\sigma_{\beta=4}^2$ variance. In all cases, we can follow the derication of Theorem~1.2 of Ref.~\cite{S_Bash13} to upper bound the error probability for transmitting $\bar m$ bits over $m$ modes averaged over a uniformly distributed codebook in a AWGN channel as 
\begin{align}
    P_{\rm err}&\leq 2^{\bar m-\frac{mN_S}{\pi\log (2)\sigma_{\beta}^2}+O(1)}\equiv P,
\end{align}
where $c_B=\frac{N_B}{1+N_B}$ and the bound holds for $MN_S=o(1)$. By setting $N_S=\frac{4\sqrt{N_B\eta^2(1+N_B\eta^2)}\delta}{(1-\eta^2)^2\sqrt{n}}$ and $M=\frac{n}{m}$, we get
\begin{align}
   \log_2P&\simeq \bar m-\frac{m}{\pi\log2}\frac{N_S}{\sigma_\beta^2}\\
   \quad &=\bar m-\frac{m}{\pi\log2}\frac{8\sqrt{N_B\eta^2(1+N_B\eta^2)}\delta}{(1-\eta^2)^2\sqrt{n}}\frac{\beta \eta^2 n/m}{1+N_B} \\
   \quad&\leq \bar m -\frac{8}{\pi\log2}c_B\beta\delta\eta^4\sqrt{n}.
\end{align}
By setting $\bar m=\frac{8}{\pi\log2}c_B\beta\delta\eta^4\sqrt{n}+\log_2\epsilon$, we get that $P\leq \epsilon$. As this calculation holds for a random codebook, it implies that there exists a specific codebook achieving this performance. Renaming $\beta_{det}=\beta$ and $\beta_{cov}=\frac{8}{\pi\log2}c_B\beta\eta^4$ concludes the proof.
\end{proof}


\section{Appendix~C: Signal-to-noise ratio of the Schr\"odinger's cat state receiver}\label{suppl:III}
Here, we quantify the performance of the circuit QED implementation of $\hat O_\tau$ in terms of the signal-to-noise ratio. \\

{\it Taylor expansion of $\hat O_\tau$}: We first expand $\hat U_\tau^\dag |e\rangle\langle e| \hat U_\tau$ with respect to the parameter $\tau$:
\begin{align}\label{expansion}
    \hat U_\tau^\dag |e\rangle\langle e| \hat U_\tau= |e\rangle\langle e|+\tau\left[\hat a_R^\dag\hat\sigma^- - \hat a_R \hat \sigma^+,|e\rangle \langle e|\right]+\frac{\tau^2}{2!}\left[\hat a_R^\dag\hat \sigma^- - \hat a_R \hat \sigma^+,[\hat a_R^\dag\hat \sigma^- - \hat a_R 
    \hat \sigma^+,|e\rangle\langle e|]\right]+o(\tau^2).
\end{align}
We have that
\begin{align}
    \left[\hat a_R^\dag\sigma^- - \hat a_R \hat \sigma^+,|e\rangle \langle e|\right]=\hat a_R^\dag \hat \sigma^-+\hat a_R\hat\sigma^+\equiv \hat E_1
\end{align}
and 
\begin{align}
    [\hat a_R^\dag\hat \sigma^- - \hat a_R \hat \sigma^+,
    \hat E_1]&=2[\hat a_R^\dag \hat \sigma^-,\hat a_R \hat \sigma^+]
    =-2|e\rangle \langle e| -2\hat \sigma_z \hat a_R^\dag \hat a_R
    \end{align}
    where $\sigma_z=|e\rangle\langle e|-|g\rangle \langle g|$ and  we have used that $[\hat \sigma^-,\hat \sigma^+]=-\hat \sigma_z$. By applying the unitary evolution $\hat S(r)\hat \sigma_x$ to each of the terms, and using the relations $\hat\sigma_x \hat\sigma^{\pm}\hat \sigma_x=\hat \sigma^\mp$, $\hat \sigma_x\hat \sigma_z\hat \sigma_x=-\hat \sigma_z$ and $\hat S(r)^\dag \hat a_R \hat S(r)=\hat a'_R$, we obtain
\begin{align}
    \hat O_\tau=|g\rangle \langle g| +\tau \hat O_{opt} + \tau^2 \hat A+o(\tau^2),
\end{align}
where $\hat A = -|g\rangle \langle g| +\hat \sigma_z \hat a'^\dag_R \hat a'_R $. This holds for $\tau^2\langle \hat A \rangle_{\rho_{\eta,\phi=0,\pi}}\ll1$. In the $N_S\ll1$, $N_B\gg1$ regime, this means roughly $\tau^2\ll 1/N_B$. \\

{\it SNR estimation:} We compute the SNR up to the second order in $\tau$, obtaining
\begin{align}
    Q_{\hat O_\tau}&\simeq \frac{\tau^2 (\langle \hat O_{opt}\rangle_{\rho_{\eta,\phi=0}}-\langle \hat O_{opt}\rangle_{\rho_{\eta,\phi=\pi}})^2}{(\lambda_+-\lambda_+^2)+\tau^2[\Delta \hat O_{opt}^2-2\langle|g\rangle\langle g|(\mathbb{I}+\hat a'^\dag_R\hat a'_R)\rangle_{\rho_{\eta,\phi=0}}-2\lambda_+^2\langle \hat A\rangle^2_{\rho_{\eta,\phi=0}}]} \\
    \quad&= Q_{\hat O_{opt}}\left[\frac{\tau^2}{a+\tau^2(1+b)}\right],
\end{align}
where $\Delta \hat O_{opt}^2=\langle\hat O_{opt}^2\rangle_{\rho_{\eta=0}}-\langle\hat O_{opt}\rangle_{\rho_{\eta=0}}^2$, $a=(\lambda_+-\lambda_+^2)/\Delta \hat O_{opt}^2$, and $b=-[2\lambda_+^2\langle \hat A\rangle^2_{\rho_{\eta,\phi=0}}+2\langle|g\rangle\langle g|(\mathbb{I}+\hat a'^\dag_R\hat a'_R)\rangle_{\rho_{\eta,\phi=0}}]/\Delta \hat O_{opt}^2$. 
Here, we have used the approximation $\langle\hat O_{opt}^2\rangle_{\rho_{\eta=0}}-\langle\hat O_{opt}\rangle_{\rho_{\eta=0}}^2\simeq \langle\hat O_{opt}^2\rangle_{\rho_{\eta,\phi}}-\langle\hat O_{opt}\rangle_{\rho_{\eta,\phi}}^2$ holding for any value of $\phi$ in  the $\eta\ll1$ limit. In the $N_S\ll1$ and $N_B\gg1$ regime, we have that $a\simeq N_S/N_B$ and $b\simeq -4N_S$. If $\tau^2\gg \frac{a}{1+b}$, then the SNR of $\hat O_\tau$ is close to the optimal one. In the $N_B\gg1$ regime, this happens whenever $\tau^2\gg N_S/N_B$. Therefore, any value $N_S/N_B\ll \tau \ll 1/N_B$ approximates $\hat O_\tau$ to the optimal observable. For instance, by setting $\tau^2=N_S/\sqrt{N_B}=\tau^{*2}$, we have that 
\begin{align}
    \frac{ Q_{\hat O_{\tau^*}}}{Q_{\hat O_{opt}}}&\simeq 1-\frac{a}{\tau^2}-b \\
    \quad&\simeq 1-\frac{1}{\sqrt{N_B}}+4N_S.
\end{align}

\section{Appendix~D: Qubit decoherence}  \label{suppl:IV}

In this section, we show how the decoherence affects the qubit measurements. We are assuming a Markovian noise described by the Lindblad operator $\mathcal{L}_D/\hbar=\frac{ \gamma}{2}\mathcal{D}[\hat \sigma_z]+\Gamma_{\uparrow}\mathcal{D}[
\hat\sigma^+]+\Gamma_{\downarrow}\mathcal{D}[\hat \sigma^-]$, where $\mathcal{D}[\hat L]\rho=(\hat L\rho \hat L^\dag -\frac{1}{2}\{\hat L^\dag \hat L,\rho\})$. In order to do so, we solve the equation $\partial_t \hat O = \mathcal{L}^\dag_D \hat O$ for different $\hat O$ defining a basis in the qubit Hilbert space, with $\mathcal{L}^\dag_D$ being the dual of $\mathcal{L}_D$. The linearity of the time-translation operator allows us to find the solution for general qubit observables. 
\begin{lemma}
We have that 
\begin{align}
e^{t\mathcal{L}_D^\dag}\hat \sigma^- &= e^{-\left[\gamma + \frac{\Gamma_{\uparrow}+\Gamma_{\downarrow}}{2}\right]t} \hat \sigma^- \label{dec1} \\ 
e^{t\mathcal{L}_D^\dag}\hat \sigma_z &= e^{-(\Gamma_{\uparrow}+\Gamma_{\downarrow})t}\hat \sigma_z +\left[1-e^{-(\Gamma_{\uparrow}+\Gamma_{\downarrow})t}\right]\frac{\Gamma_{\uparrow}-\Gamma_{\downarrow}}{\Gamma_{\uparrow}+\Gamma_{\downarrow}}\mathbb{I}. \label{dec2}
\end{align}
\end{lemma}
\begin{proof}
The relations can be derived by simply checking the action of the decoherence generators on the operator of interest. Notice that $\mathcal{L}_D^\dag = \frac{\gamma}{2}\mathcal{D}[\hat \sigma_z]^\dag+\Gamma_{\uparrow}\mathcal{D}[\hat \sigma^+]^\dag+\Gamma_{\downarrow}\mathcal{D}[\hat \sigma^-]^\dag$,
where $\mathcal{D}[L]^\dag \hat O= L^\dag \hat O L -\frac{1}{2} \{L^\dag L ,\hat O\}$. We have that 
\begin{align}
\mathcal{D}[\hat \sigma_z]^\dag \hat \sigma^- &= -2\hat\sigma^- \\ 
\mathcal{D}[\hat \sigma^+]^\dag \hat \sigma^- &= -\frac{\hat \sigma^-}{2} \\ 
\mathcal{D}[\hat \sigma^-]^\dag \hat \sigma^- &= -\frac{\hat \sigma^-}{2},
\end{align}
from which Eq.~\eqref{dec1} follows trivially using $e^{t \mathcal{L}_D^\dag}=\sum_{k=0}^\infty \frac{t^k \mathcal{L}_D^{\dag k}}{k!}$. We have that 
\begin{align}
\mathcal{D}[\hat \sigma_z]^\dag \hat\sigma_z &= 0 \\ 
\mathcal{D}[\hat\sigma^+]^\dag \hat\sigma_z &= 2|g\rangle\langle g| \\ 
\mathcal{D}[\hat\sigma^-]^\dag \hat\sigma_z &= -2|e\rangle\langle e|,
\end{align}
from which it follows that $\mathcal{L}_D^\dag \hat\sigma_z = -(\Gamma_{\uparrow}+\Gamma_{\downarrow})\hat\sigma_z + (\Gamma_{\uparrow}-\Gamma_{\downarrow})\mathbb{I}$. By using that $\mathcal{L}_D^\dag \mathbb{I}=0$, we infer that $\mathcal{L}_D^{\dag k}\hat \sigma_z= (-1)^k(\Gamma_{\uparrow}+\Gamma_{\downarrow})^k\hat\sigma_z + (-1)^{k-1}(\Gamma_{\uparrow}+\Gamma_{\downarrow})^{k-1}(\Gamma_{\uparrow}-\Gamma_{\downarrow})\mathbb{I}$, $k=1,2,\dots$ Eq.~\eqref{dec2} follows trivially.
\end{proof}
Therefore, we have that $T_1=(\Gamma_{\uparrow}+\Gamma_{\downarrow})^{-1}$ and $T_2=\left(\gamma+\frac{\Gamma_{\uparrow}+\Gamma_{\downarrow}}{2}\right)^{-1}$.
As $\hat O_{opt}=\sigma^-(\lambda_+\hat a+\lambda_-\hat a^\dag)+c.c.$ is linear in $\hat\sigma^-$ and $\sigma^+$, we have that $Q_{\hat O_{opt}}^{dec}/Q_{\hat O_{opt}}=e^{-2t/T_2}$, which is Eq.~\eqref{noisySNR} of the main text.


\begin{thebibliography}{100}
\bibitem{footnote} The squeezing operation on the mode $\hat a_R$ is defined as $\hat S(\xi)=\exp\left(-\frac{\xi^2}{2}\hat a_R^{\dag 2}+\frac{\xi^{*2}}{2} \hat a_R^2\right)$, where $\xi=re^{-i\phi}$ is the squeezing parameter.
\end{thebibliography}
\end{document}